%% file: gradient_final.tex
\documentclass{article}

\usepackage{microtype}
\usepackage{graphicx}
\usepackage{subfigure}

\usepackage{hyperref}
\usepackage{multirow}
\usepackage{stfloats}

\usepackage{booktabs} % For formal tables

\newcommand{\vValue}{\psi}
\newcommand{\indicator}[1]{\textbf{1}_{\left[{#1}\right]}}
\usepackage{empheq}
\usepackage{amsmath,amssymb,amsthm,boldline}
\usepackage[utf8]{inputenc}
\usepackage[normalem]{ulem}
\newcommand{\Unif}{\mathcal{U}}
\DeclareMathOperator{\argmax}{argmax}

\newcommand{\shadingFunc}{\beta}

\newcommand{\Expb}[1]{\mathbb{E}\left(#1\right)}

\newcommand{\eps}{\epsilon}
\newcommand{\tendsto}{\rightarrow}

\newcommand{\trsp}{'}%for transpose
\newtheorem{lemma}{Lemma}
\newtheorem{corollary}{Corollary}
\newtheorem{theorem}{Theorem}
\usepackage[accepted]{icml2019}

\icmltitlerunning{Learning to bid in revenue-maximizing auctions}

\begin{document}

\twocolumn[
\icmltitle{Learning to bid in revenue-maximizing auctions}



\icmlsetsymbol{equal}{*}

\begin{icmlauthorlist}
\icmlauthor{Thomas Nedelec}{criteo,ens}
\icmlauthor{Noureddine El Karoui}{berkeley,criteo}
\icmlauthor{Vianney Perchet}{ens,criteo}
\end{icmlauthorlist}

\icmlaffiliation{criteo}{Criteo AI Lab}
\icmlaffiliation{ens}{CMLA, ENS Paris Saclay}
\icmlaffiliation{berkeley}{UC, Berkeley}

\icmlcorrespondingauthor{Thomas Nedelec}{nedelec@cmla.ens-cachan.fr}

\icmlkeywords{Machine Learning, Market design, Auctions, Optimization}

\vskip 0.3in
]



\printAffiliationsAndNotice{}  % leave blank if no need to mention equal contribution

\begin{abstract}
We consider the problem of the optimization of bidding strategies in prior-dependent revenue-maximizing auctions, when the seller fixes the reserve prices based on the bid distributions. Our study is done in the setting where one bidder is strategic. Using a variational approach, we study the complexity of the original objective and we introduce a relaxation of the objective functional in  order to use gradient descent methods. Our approach is simple, general and can be applied to various value distributions and revenue-maximizing mechanisms. The new strategies we derive yield  massive uplifts compared to the traditional truthfully bidding strategy. 
\end{abstract}

\section{Introduction}
Modern marketplaces like Uber, Amazon or Ebay enable sellers to fine-tune their selling mechanism by reusing their large number of past interactions with consumers. In the online advertising or the electricity markets, billions of auctions are occurring everyday between the same bidders and  sellers. Based on the data gathered, different approaches  learn complex mechanisms maximizing the seller revenue  \cite{conitzer2002complexity,OstSch11,paes2016field,Golrezaei2017}.

Most of the literature has focused on the auctioneer side \cite{milgrom2018artificial}. Algorithms focused on the bidder's standpoint to enable them to be strategic against any smart data-driven selling mechanisms are lacking. These algorithms should ideally strengthen the balance of power driving the relationship between buyers and sellers. Our main objective is to exhibit simple robust algorithmic procedures that take advantage of various data-dependent revenue-maximizing mechanisms. This represents a big step forward in understanding possible strategic behaviors in revenue maximizing auctions. This is a new argument supporting the Wilson doctrine \cite{Wil87} claiming that data-dependent revenue maximizing algorithms are not robust to strategic bidders.
\subsection{Framework}
In the early stage of the market design literature (see, e.g., \citet{Myerson81}),  a typical underlying assumption is that the bidders' value distributions were commonly known to the seller and other bidders. This can be justified if different group of bidders with the same value distribution are interacting successively with one seller. In the  aforementioned modern applications, the same bidders have billions of interactions everyday with the seller. Even if the latter does not know the value distribution beforehand, it might use in many cases the past bid distributions as proxies of  value distribution.  

Several mechanisms  based on the value distribution of  bidders have already been introduced. We will focus on the lazy second price auction with personalized reserve price \citep{paes2016field}, the Myerson auction \citep{Myerson81}, the eager version of the second price auction and the boosted second price auction \citep{Golrezaei2017}.  When repeating these  auctions (every day, or every milli-second, depending on the context) and if the bidder is \textit{myopic}, i.e optimizing per stage and not long-term revenue, it is optimal to bid truthfully at each auction. So with myopic bidders, bids and values  have  the same  distribution and the seller can  design optimally the mechanism based on the former. 

Non-myopic bidders  optimize their \textit{long-term expected utility} taking into account that their current strategy will imply a certain mechanism (for instance a specific reserve price) in the future. More precisely, we will consider the following steady state analysis. Assume the valuations of a bidder $v_i \in \mathbb{R}$ are drawn from a specific distribution $F_i$;  a bidding strategy is a mapping $\beta_i$ from $\mathbb{R}$ into  $\mathbb{R}$ that indicates  the actual bid $B_i=\beta_i(v_i)$ when  the value is $v_i$. As a consequence, the distribution of bids $F_{B_i}$ is  the push-forward of $F_i$ by $\beta_i$. In the steady state, the seller uses the distributions of bids $F_{B_i}$ to choose a specific auction mechanism $\mathcal{M}(F_{B_i})$  among a given class of mechanisms $\mathcal{M}$. The objective of \textit{a long-term strategic bidder} is to find her strategy $\beta_i$ that maximizes her expected utility when $v_i \sim F_i$, she bids $\beta_i(v_i)$ and the induced mechanism is $\mathcal{M}(F_{B_i})$. This steady-state  objective is particularly relevant in modern applications as most of the data-driven selling mechanisms are using large batches of bids as examples to update their mechanism.

In terms of game theory, these interactions are a game between the seller - whose strategy is to pick a mechanism design that  maps bid distributions to reserve prices - and the bidders - who chose bidding strategies. Our overarching objective is to derive the best-response, for a given bidder $i$,  to  the strategy of the seller (i.e., a given mechanism) and the strategies of the other bidders (i.e., their bid distributions).

\subsection{Contributions}

Our main contributions are the following. We first introduce the optimization problem that strategic bidders are facing when the seller is optimizing personalized reserve prices based on their bid distributions. A straightforward optimization can fail because the objective is  discontinuous as a function of the bidding strategy. 

To circumvent this issue, we introduce a new relaxation of the problem which is stable to local perturbations of the objective function and computationally tractable and efficient. We numerically optimize this new objective through a simple neural network and get very significant improvements in bidder utility compared to truthful bidding. We also provide a theoretical analysis of \textsl{thresholded strategies} (introduced in \citet{nedelec2018thresholding}) and show their (local) optimality as improvements of bidding strategies with non-zero reserve value. 

For the Myerson auction, the strategies  learned by the model can be independently proved to be optimal. We apply the approach to other auction settings such as boosted second price or eager second price with monopoly price. We report massive uplifts compared to the traditional truthful strategy advocated in all these settings. Our simple approach can be plugged in any modern bidding algorithms learning distribution of the highest bid of the competition and we test it on other classes of mechanism without any known closed form optimal bidding strategies. We finally provide the code in PyTorch that has been used to run the different experiments. This approach opens avenues of research for designing good bidding strategies in many data-driven revenue-maximizing auctions.

\subsection{Related work}

Starting with the seminal work of \citet{Myerson81}, a rich line of work indicates the type of auctions that is revenue-maximizing for the seller. In the case of symmetric bidders \cite{Myerson81}, one revenue maximizing auction is a second price auction with a reserve price equal to the monopoly price, i.e, the price $r$ that maximizes $r(1-F(r))$. However, in most applications, the symmetric assumption is not satisfied \cite{Golrezaei2017}. In the asymmetric case, the Myerson auction is optimal \cite{Myerson81} but is difficult to implement in practice \cite{morgenstern2015pseudo}. In this case, a second price auction with a well-chosen vector of reserve prices guarantees at least one-half of the optimal revenue \cite{Hartline2009}. 
 
In modern markets, some bidders are myopic  simply because truthful bidding is a simple strategy to implement. Receiving truthful bid enables sellers to design various revenue maximizing auctions. \cite{conitzer2002complexity} has therefore been interested in the automatic mechanism design that fine tunes mechanism based on some examples of bids. This work was extended recently in \cite{dutting2017optimal} with the use of deep learning. In \cite{OstSch11,medina2014learning,paes2016field}, it is shown specifically how to learn the optimal reserve prices in the lazy second price auction. This practice was theoretically addressed by \cite{cole2014sample,huang2018making,devanur2016sample} looking at the sample complexity of a large class of auctions assuming an oracle offering iid examples of the value distribution.

However, it is quite intuitive that non-myopic bidders should not bid truthfully. Robustness to strategic bidders has been studied in \cite{balseiro2017dynamic,kanoria2017dynamic,epasto2018incentive}. A potential limitation of this type of approach is that it is either assumed that all bidders have the same value distribution (or up to $\varepsilon$ for some specific metric on distributions) or that there is a very large number of bidders and a global mechanism designed so that any of them has no incentive to bid untruthfully. In \cite{ashlagi2016sequential}, an involved mechanism was designed that keeps the incentive compatibility property even if the seller is learning on former bids of the bidders. 

None of these papers have exhibited  optimal strategies that can be used when the seller is optimizing her mechanism based on past bids. This strategic behavior has been studied for posted price with one bidder and one seller \cite{mohri2015revenue}. An independent line of work has focused on learning to bid when the value is not known to the bidders \cite{weed2016online,feng2018learning}. Some Bayes-Nash equilibria corresponding to games where bidders can choose their bid distribution were designed  \cite{tang2016manipulate,abeille2018explicit} with some derivations of seller revenue and bidders utility at these equilibria. However, no strategies corresponding to these equilibria were provided in the general case. Our work is finally strongly related to \cite{nedelec2018thresholding} where a new class of shading strategies for second price auctions with personalized reserve price is proposed. Our new optimization pipeline is very general and enables bidders to learn good bidding strategies in multiple settings and for any value distribution. 
 
\section{The bidder's optimization problem }
We introduce in this section the optimization problem, starting with the lazy second price auction with personalized reserve prices (formalized below).
\subsection{Notations and setting}
To describe precisely our approach, we use the traditional setting of auction theory (see e.g. \citet{krishna2009auction}). Recall that $F_{i}$ is the value distribution of bidder $i$ and $\beta_i : \mathbb{R} \to \mathbb{R}$ her strategy that maps values to bids. The corresponding distribution of bids is then $F_{B_i} = \beta_i \sharp F_i$, the push-forward of $F_i$ w.r.t.\ $\beta_i$. In the steady-state, we assume that the seller has the perfect knowledge of each bid distribution $F_{B_i}$.  Notice that we have implicitly identified the distribution $F_i$ (resp.\ $F_{B_i}$) with its cumulative distribution function (cdf) and use both terms exchangeably. We use $f_i$ (resp.\ $f_{B_i}$) for the corresponding probability density function (pdf).

 For the sake of simplicity,  let us first consider a lazy second price auction \citet{krishna2009auction}.  We recall that in this auction each bidder has a personalized reserve price.  The item is attributed to the highest bidder, if she clears her reserve price, and not attributed otherwise; the winner then pays the maximum between the second highest bid and her reserve price.  It is known that the optimal reserve price of bidder $i$ is  her monopoly price equal to  $\argmax_r r(1-F_{B_i}(r))$, or equivalently\footnote{at least for regular distributions, i.e., when $\psi$ is non-decreasing} to $\psi_{B_i}^{-1}(0)$, where $\psi_{B_i}$ is the usual  virtual value function defined as
\begin{equation*}
\psi_{B_i}(b) = b - \frac{1-F_{B_i}(b)}{f_{B_i}(b)}\;.
\end{equation*}
As a consequence, it is natural to assume that the strategy of bidder $i$ does not impact the strategy of other bidders (that can be either myopic or not) and from now on, we assume that bids are independent.

\subsection{A variational approach}
A fundamental result in auction theory is the Myerson lemma  \cite{Myerson81}. It expresses the expected payment of a bidder depending on her virtual value and the value distribution of the competition.  An important notation is $G_i$, the cdf of the maximum bid of  players other than~$i$; obviously, if the other bidders are  truthful, $G_i$ is the distribution of the maximum value of the other bidders.

\begin{lemma}[Integrated version of the Myerson lemma]\label{Myerson_lemma}
%Let bidder $i$ have value distribution $F_i$ and denote $\beta_i$ her strategy, $F_{B_i}$ the induced distribution of bids and $\psi_{B_i}$  the corresponding virtual value function.  Suppose that  bids of player $i$ are independent of the bids of the other bidders and 
In a lazy second price auction with personalized reserve price $r_i$,  the payment of bidder $i$ with continuous strategy $\beta_i$ is
\begin{equation*}
\Pi(\beta_i) = \mathbb{E}_{B_i \sim F_{B_i}}\bigg(\psi_{B_i}(B_i)G_i(B_i)\textbf{1}(B_i \geq r_i)\bigg)\;.
\end{equation*}
\end{lemma}
\begin{proof}
The proof is similar to the original one \cite{Myerson81}, see also \cite{krishna2009auction}, so we do not spell it out. 
It consists of using Fubini's theorem and integration by parts to transform the standard form of the seller revenue, i.e.
$$
\mathbb{E}_{B_i\sim F_{B_i},X_j\sim F_{B_j}}\bigg(\max_{j\neq i}(B_j,r)\indicator{B_i\geq \max_{j\neq i}(B_j,r)}\bigg)
$$
into the above equation.  It then suffices to work along the lines mentioned above with  $Y_i=\max_{j\neq i}B_j$ and realize that $i$'s expected payment can be written as
$\displaystyle 
\mathbb{E}_{B_i\sim F_{B_i},Y_i\sim G_i}\bigg(\max(Y_i,r)\indicator{B_i\geq \max(Y_i,r)}\bigg)\;.
$
\end{proof}
% The bidder payment is equal to the area under the curve (we removed the dependence on G for the sake of clarity). The seller has no incentive to set the reserve price lower than $\psi_i^{-1}(0)$ since it decreases her total expected payment (because of the negative contribution of the red area). It is also clear that $r \ge \psi_i^{-1}(0)$ is suboptimal since it results in lost  revenue for the seller.

In lazy second price auction,  the seller chooses as reserve price the monopoly price corresponding to the bid distribution of bidder $i$. In this case, Lemma \ref{Myerson_lemma} implies that the expected payment of bidder $i$ is equal to 
\begin{equation*}
\Pi(\beta_i) = \mathbb{E}_{B\sim F_{B_i}}\bigg(\psi_{B_i}(B)G_i(B)\textbf{1}(B \geq \psi_{B_i}^{-1}(0))\bigg)\;.
\end{equation*}
In order to simplify the computation of the expectation and remove the dependence on $\beta_i$, this expected payment can be rewritten in the space of values, by introducing
\begin{equation*}
h_{\beta_i}(x) \triangleq \psi_{F_{B_i}}(\beta_i(x)) \; ,
\end{equation*}
and noting the equivalent following formulation
\begin{equation*}
\Pi(\beta_i) = \mathbb{E}_{X_i \sim F_{i}}\bigg(h_{\beta_i}(X_i)G_i(\beta_i(X_i))\textbf{1}(X_i \geq x_{\beta_i})\bigg)\;,
\end{equation*}
where $x_{\beta_i} =  h_{\beta_i}^{-1}(0)$ when $h_{\beta_i}$ is  increasing. We call it  \emph{the reserve value}, as it is the smallest value above which the seller accepts all bids from bidder $i$. 

The expected utility can be derived as a function of $\beta_i$ as 
\begin{equation}\label{equ:utility}
U(\beta_i) = \mathbb{E}_{X_i \sim F_{i}}\bigg((X_i-h_{\beta_i}(X_i))G_i(\beta(X_i))\textbf{1}(X_i \geq x_{\beta_i})\bigg)\;.
\end{equation}
Finally, we remark that if $\beta_i$ is increasing and differentiable, $h_{\beta_i}$ verifies a simple first order differential equation.
\begin{lemma}\label{definition_psi}
Suppose  $\beta_i$ is increasing and differentiable then
\begin{equation}\label{eq:ODEPhiG}
h_{\beta_i}(x_i) = \psi_{F_{B_i}}(\beta_i(x))  =\beta_i(x)-\beta_i'(x)\frac{1-F_{i}(x)}{f_{i}(x)}\;.
\end{equation}
\end{lemma}
\begin{proof}
$\psi_{F_{B_i}}(b) = b - \frac{1-F_{B_i}(b)}{f_{B_i}(b)}$ with $F_{B_i}(b) = F_{i}(\beta_i^{-1}(b))$ and $f_{B_i}(b) = f_{i}(\beta_i^{-1}(b)/\beta_i'(\beta_i^{-1}(b)$. Then, 
$h_{\beta_i}(x) = \psi_{B_i}(\beta_i(x)) = \beta_i(X) - \beta_i'(X)\frac{1-F_{i}(X)}{f_{i}(X)}\;.$
\end{proof}

%With this approach, the expected utility of the strategic bidder can be expressed directly as a function of $\beta_i$. 
%\begin{equation*}
%U(\beta_i) = \mathbb{E}_{X_i \sim F_{i}}\bigg((X_i-h_{\beta_i}(X_i))G_i(\beta(X_i))\textbf{1}(X_i \geq x_{\beta_i})\bigg)\;.
%\end{equation*}
%with $h_{\beta_i}(x) = \beta_i(X) - \beta_i'(X)\frac{1-F_{i}(X)}{f_{i}(X)}\;.$ and $\beta_i$ increasing and differentiable.

If we consider only monotonically increasing differentiable strategies, and we denote by  $\mathcal{I}$ the class of  such functions, the problem of the strategic bidder is therefore to solve  $ \max_{\beta \in \mathcal{I}} U(\beta) $ with $U$ defined in Equation \eqref{equ:utility}. This equation is crucial, as it indicates that optimizing over bidding strategy can be reduced to finding a distribution with a well-specified virtual value $h(\cdot)$. A crucial difference between the long term vision and the classical, myopic (or one-shot)  auction theory is that bidders also maximize expected utility. They might therefore be willing to sometime over-bid (incurring a negative utility at some specific auctions) if this reduces their reserve price. Indeed, having a lower reserve price increases the revenue of many other auctions. Lose small to win big. This reasoning is possible as there exist multiple interactions between  bidders and seller, billions every day in the case of online advertising. 

\subsection{Discontinuity of the objective}

In the previous section, we assumed  the reserve value was defined as $h_{\beta_i}^{-1}(0)$, which is well defined only if $h_{\beta_i}$ is increasing. This condition is complicated to ensure as, for instance, restricting the  strategies to be increasing does not provide any guarantee on $h_{\beta_i}$. If the later is not increasing, then the function $r(1-F_{B_i}(r))$ that the seller maximizes might have several local optima, as illustrated with a specific  bid distribution  in Figure \ref{fig:utility_bidder}. We mention here that this distribution actually arises during our numerical optimization using first order splines as described in the next section. 

\begin{figure}[h!]
\center
 \begin{tabular}{c}
 \includegraphics[width=0.30\textwidth, height=0.19\textheight]{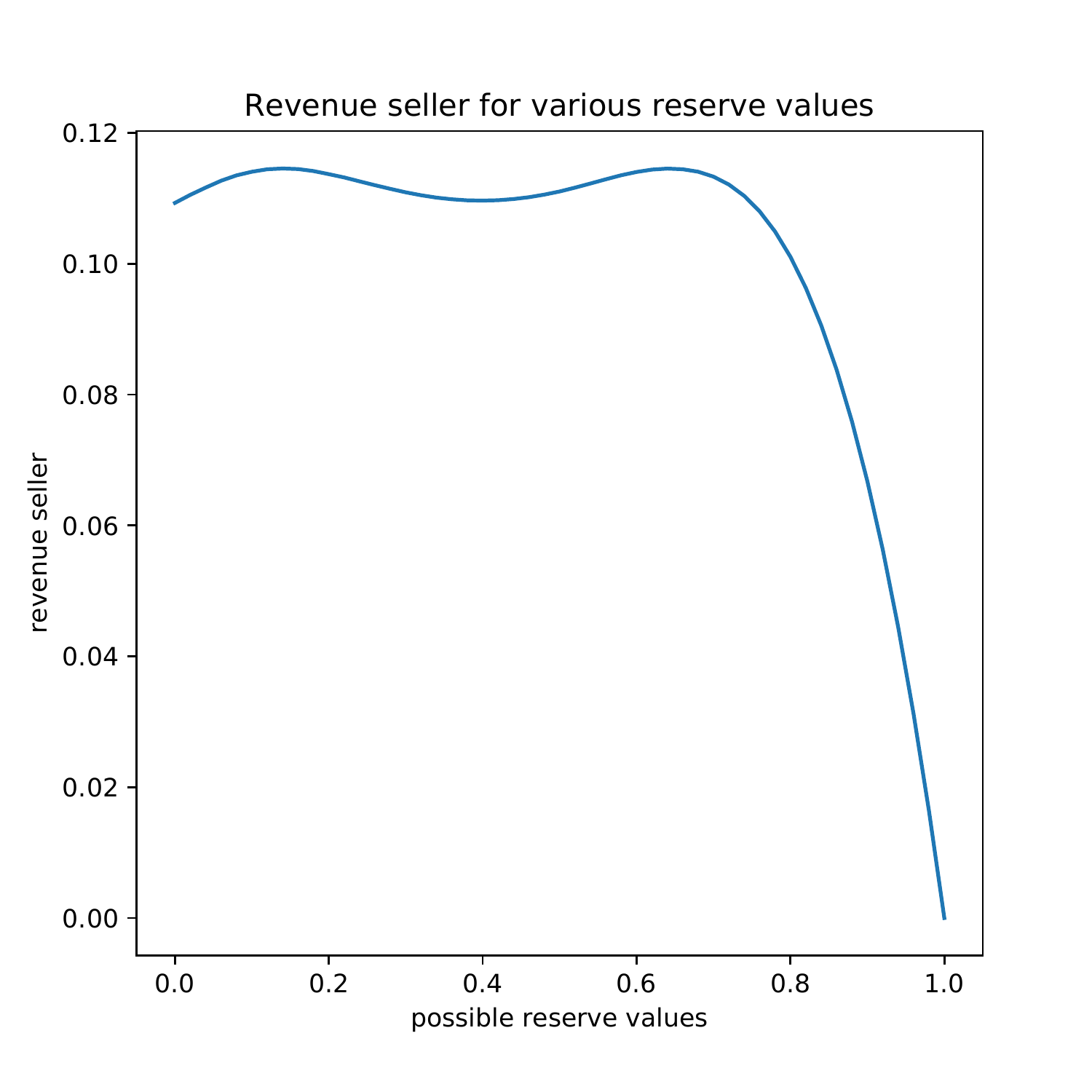}
  \end{tabular}
 \caption{\textbf{Revenue of the seller as a function of the reserve value.} This shape of revenue by running the first order spline method described in Section \ref{splines_section} . For this  distribution, there exists two local optima that are equivalent in terms of revenue for the seller but dramatically change the utility of the strategic bidder. }
 \label{fig:utility_bidder}
 \end{figure}

The fact that $r(1-F_{B_i}(r))$  is not always strictly concave implies that the set of maximizer is not continuous but only upper hemi-continuous; stated otherwise, the reserve value $x_{\beta_i}$ can ``jump'' from a small to high value with an arbitrarily small change in the bidding strategy. In the example of Figure \ref{fig:utility_bidder}, the reserve value switches from 0.18 to 0.58. As a consequence, the expected utility of the bidder, which is another function depending on $x_{\beta_i}$, might also jumps erratically. In the same example, the lower bound of integration increases from 0.18 to 0.58, so that the overall integral decreases  from 0.14 to 0.09. This discontinuity makes the optimization of the real objective difficult. 

\subsection{An attempt with first-order splines}
\label{splines_section}
A natural question is  whether the buyer can compute shading strategies  numerically. A first approach is to look back at the gradient of the bidder's utility in the direction of a certain function $\rho$, i.e.,  the directional derivative, that can be computed by elementary calculus.
\iffalse
\begin{align*}
\left.\frac{\partial}{\partial t} U_i(\beta_t)\right|_{t=0}&=\Expb{\left\{g(\beta(X_i))[X_i-\beta(X_i)]\right\}\rho(X_i)\indicator{X_i\geq x_\beta}}\\
&+G(\beta(x_{\beta}))\left[\frac{x_{\beta}f(x_{\beta})}{h_{\beta}^{'}(x_\beta)}-(1-F(x_{\beta}))\right]\rho(x_{\beta})\\
&-\frac{\rho'(x_{\beta})x_{\beta}(1-F(x_{\beta}))G(\beta(x_{\beta}))}{h_{\beta}^{'}(x_\beta)}\;.
\end{align*}
\fi
and to look at shading function expressed in a specific basis as$$
\beta(x)=\sum_{i=1}^N c_i(\beta) f_i(x)\;,
$$
and try to optimize over $c_i$. It would be also  quite natural to do  an isotonic regression and optimize over non-decreasing functions directly; this approach is tackled later on. 

\paragraph{A natural basis} Splines (see e.g. \cite{htfElementsStatLearning} for a practical introduction) are a natural candidate for the function $f_i$'s. In particular, first order splines are piecewise continuous functions, hence evaluating derivatives is trivial and it is easy to account in the formula above for the finitely many discontinuities of the derivative that will arise. If $\xi_k$'s are given knots, first order splines are the functions $$
f_1=1, f_2(x)=x, f_{k+2}(x)=(x-\xi_k)_+=\max(x-\xi_k,0)\;.
$$
Higher order splines could of course also be used. 
\begin{lemma}
As described above, the optimal shading problem can be numerically approximated using steepest descent by a succession of linear programs, provided the non-decreasing constraint on $\beta$ can be written linearly in $c_i$. This is of course the case for 1st order spline.
\end{lemma}
\begin{proof}
After the function is expanded in a basis, the functional gradient  becomes a standard gradient, and the shading function can be improved with a steepest descent.  If the reserve value is not one of the knots, the gradient above is easy to compute:  each step of the optimization requires to solve a constrained LP to ensure that the solution is increasing. 

For 1st order splines,  the derivative is constant between knots, thus checking that $\beta'(\cdot)\geq 0$  amounts to check finitely many linear constraints and so is amenable to an LP. 
\end{proof}

%This approach amounts to an exploration of parts of the space of shading functions. Since 

The objective is not even continuous, though differentiable in a large part of the parameter space. The optimization problem is hard. In our experiments, we got significant improvement over bidding truthfully by using the above numerical method. However, we encountered the discontinuities of the optimization problem described above: our numerical optimizer got stuck at shading functions around which the reserve value was very unstable, which corresponds to revenue curves for the seller with several distant (approximate) local maxima: a small perturbation in function space does not induce much loss of the revenue on the seller side, but can have a huge impact on the reserve value and hence the buyer revenue. Note that in our numerical experiments we did not enforce the non-decreasing-constraint on $\beta$ but ended up with solutions that were non-decreasing. More details on this approach are provided in Appendix \ref{appendix_splines}.

This is precisely the reason why, in the next section, we introduce a relaxation of the problem that is easier to optimize and with the same solutions as our initial objective. We also change the class of shading functions we consider and use neural networks to fit them. Before we describe these experiments, we provide some theory for the problem of optimizing buyer revenue in lazy second price auctions. 

\section{Theory and a relaxation of the problem enabling the use of gradient descent}
\subsection{The family of optimal extensions of a strategy}
\label{theory}
In the context of lazy second price auctions, any increasing and continuous bidding strategy  $\beta$ whose reserve value $r$ is not 0 can be improved with $\beta^*(x)= \frac{\beta(r)(1-F(\underline{r}))}{1-F(x)}$ on $[0,r]$ where $\beta(\underline{r})=r$,  i.e.  by thresholding the virtual value below the current reserve value and keeping $\beta^*=\beta$ on $(r,+\infty)$.  Indeed, Lemma \ref{definition_psi} yields that $h_{\beta^*}(x)=0$ on $[0,\underline{r}]$ and $h_{\beta^*}=h_{\beta}$ elsewhere. So the seller is indifferent between setting the reserve price anywhere in $[0,\underline{r}]$ and we might assume she picks 0 (if she is \textsl{welfare benevolent}, or it is always possible to give an $\varepsilon$-incentive to pick 0, for $\varepsilon$ arbitrarily small).
 %in the case where the seller is welfare benevolent, i.e. sets her reserve price at the smallest price that brings her maximal revenue in case there are several such points.
According to Myerson's Lemma, the strategy $\beta^*$ generates the same payment as $\beta$, so the revenue of the seller coming from this bidder is unchanged. On the other hand, that bidder wins more auctions with this new strategy,  hence it improves her revenue and thus her expected utility.

In this subsection, we  address the question of whether the strategy $\beta^*$, which is simple and robust can be improved for the bidder. Our previous argument already shows that any  improvement would be a strategy with 0 reserve value. 

Differentiating Equation \eqref{eq:ODEPhiG} yields
%. We call 
%$$
%h_\beta(x)=\beta(x)-\beta'(x)\frac{1-F(x)}{f(x)}=\psi_B(\beta(x))\;.
%$$
%We note that, by simply differentiating the right hand side below,  
$$
f(x)h_\beta(x)=(\beta(x)(F(x)-1))'\;.
$$
Let us denote by $r$ the current reserve price; we rewrite the family of bidding strategies $\beta$ with reserve value at 0 as elements of the following  constraint set:
\begin{gather*}
\Expb{\psi_B(\beta(X))G(\beta(X))\indicator{r_0\leq X\leq r}}\leq 0 \;, \forall 0\leq r_0 \leq r\;, \\
\Expb{\psi_B(\beta(X))G(\beta(X))\indicator{0\leq X\leq r}}=0\;.
\end{gather*}
For all those strategies, the seller revenue is maximal for the reserve value $r_{opt}=0$, and hence under the assumption of welfare benevolence, the seller will accept all  bids of the bidder. It is also clear that this set of constraints define all possible strategies with reserve value 0. 

The strategy $\beta$ (which is increasing and continuous, say) that maximizes the revenue of the bidder corresponds to
$$
\max_\beta \Expb{(X-\psi_B(\beta(X)))G(\beta(X))\indicator{X\geq 0}}
$$
under the constraints that 
\begin{gather*}
\mathfrak{g}_{r_0}(\beta)=\Expb{\psi_B(\beta(X))G(\beta(X))\indicator{r_0\leq X\leq r}}\leq 0 \\
\mathfrak{g}_{0}(\beta)=\Expb{\psi_B(\beta(X))G(\beta(X))\indicator{0\leq X\leq r}}=0\;.
\end{gather*}
Let us limit ourselves to not changing our strategy beyond $r$, e.g. by bidding truthfully beyond $r$. Then we effectively need to maximize 
$$
\max_\beta F(\beta)=\Expb{(X-h_\beta(X))G(\beta(X))\indicator{0\leq X\leq r}}\;.
$$
with the continuity constraints that $\beta(r)=r$. The constraints can be rewritten into 
\begin{align*}
\mathfrak{g}_{r_0}(\beta)&=-\Expb{\psi_B(\beta(X))G(\beta(X))\indicator{0\leq X\leq r_0}}
\\&=-\Expb{h_\beta(X)G(\beta(X))\indicator{0\leq X\leq r_0}}\;.
\end{align*}
along with $\mathfrak{g}_{r}(\beta)=0$. We call those strategies continuation strategies as they extend the bidding below the current reserve price/value.

\textbf{Remark~:} in this class of feasible strategies, the optimal reserve value for the seller is zero. So the discontinuities of the objective function in the broader class of strategies considered before, which stemmed from discontinuities of the reserve value as a function of the shading function, are not anymore problematic. %Indeed our objective is now very smooth. 

The following theorem states one of our main results.
\begin{theorem}\label{thm:localAndGlobalOptimalityofThresholding}
Let  $F$, $1/(1-F)$ and $G$ be differentiable on $[0,r]$. Suppose  that the virtual value $\psi_F$ is such that  $\psi_F(x)\leq 0$ on $[0,r]$. We consider increasing shading functions $\beta$ on $[0,r]$ with $\beta(r)=r$. 

Thresholding, i.e. using $\beta^*(x)=r(1-F(r))/(1-F(x))$ for $0\leq x\leq r$ is locally optimal among continuation strategies for which $\beta$ is differentiable on $[0,r]$, provided $G(\beta^*(x))>0$ on $[0,r]$. It is also  locally optimal among $\beta$'s such that $\mathfrak{g}_{r}(\beta)$ is differentiable as function of $r$. 

Furthermore, if $r<1$ and $G(x)=\min(x,1)$, i.e. the competition's distribution is Uniform$[0,1]$, then thresholding is globally optimal among functions that are bounded by 1 and differentiable.
\end{theorem}
\textbf{Sketch of proof~:} the proof consists in keeping track of the slack function $h(r)=\mathfrak{g}_{r}(\beta)$, rewriting locally feasible $\beta$'s as functions of $h$ through differential equation manipulations and finally comparing their revenue and showing that the optimal $h$ is zero for our objective. This requires somewhat lengthy and delicate manipulations. In the case of $G(x)=\min(x,1)$, we are able to write all feasible $\beta$'s as a function of $h$ and carry out the program globally.

We note that we did not require in our analysis that our optimization be limited to non-decreasing functions; it turns out that our local optima are optimal in larger class of functions. 

\subsection{One relaxation of the objective} 
Instead of computing the exact reserve value in the definition of the expected utility of the bidder, we introduce a relaxation $U_r$ of the objective corresponding to : 
\begin{equation}\label{equ:relaxation}
U_r(\beta_i) = \mathbb{E}\big((X_i-h_{\beta_i}(X_i))G_i(\beta(X_i))\textbf{1}_{[h_{\beta_i}(X_i) \geq 0]}\big)\;.
\end{equation}
We replaced $\textbf{1}_{[X_i \geq x_{\beta_i}]}$ by $\textbf{1}_{[h_{\beta_i}(X_i) \geq 0]}$. 
This relaxation avoids to compute the reserve value at each step of the gradient descent and remove most of the discontinuities of the previous objective. We now prove that the function maximizing Equation. \ref{equ:relaxation} has non-negative virtual value. The value of the relaxation objective at its optimum is equal to the one in the strategic bidder problem.
\begin{theorem}
 If an increasing and differentiable function $\beta_i$ is maximizing 
 \begin{equation*}
U_r(\beta_i) = \mathbb{E}\big((X_i-h_{\beta_i}(X_i))G_i(\beta_i(X_i)))\textbf{1}_{[h_{\beta_i}(X_i) \geq 0]}\big)\;,
\end{equation*}
it has non-negative virtual value, a reserve value equal to zero and $U_r(\beta_i) = U(\beta_i)$ with  
\iffalse Furthermore, if $G(x)=\min(x,1)$, i.e. the competition's distribution is Uniform[0,1], the relaxation is tight and the maximizer of Equation \ref{equ:relaxation} is the same as the optimizer of
\fi
\begin{equation*}
U(\beta_i) = \mathbb{E}\big((X_i-h_{\beta_i}(X_i))G_i(\beta_i(X_i)))\textbf{1}_{[X_i \geq x_{\beta_i}]}\big)\;.
\end{equation*}

\end{theorem}
\begin{proof} We use the fact that if $h_{\beta}(x) < 0$ on a certain interval [a,b], we can find a new strategy $\beta^+$ with higher $U_r$.  Let us consider the rightmost interval [a,b] where $h_{\beta}(x) < 0$. On $[b,+\infty]$, $\beta^{+} = \beta$. Then on [a,b], $\beta^{+}(x)=\beta(b)(1-F(b)/(1-F(x))$. $\beta^{+}$ verifies $h_{\beta^{+}}(x) = 0$ on $[a,b]$. Then if we denote $T = \beta^{+}(a)$, we define  $\beta^{+}$ on [0,a] as $\beta^{+}(x)=\beta(x)+(T-\beta(a))(1-F(a))/(1-F(x))$. We have $h_{\beta^{+}} = h_{\beta}$ on [0,a]. $\beta^{+}$ is continuous. With $f(x)h_\beta(x)=(\beta(x)(F(x)-1))'$, we see that $\beta(1-F)$ is non-decreasing on [a,b]. Hence $\beta^{+}(a)\geq \beta(a)$. Therefore, $\forall x, \beta^{+}(x) \geq \beta(x)$ and $G(\beta^{+}(x)) \geq G(\beta(x))$. Hence, $U_r(\beta^*) \geq U_r(\beta)$. Then, we tackle the next interval where $h_{\beta}(x) < 0$ by doing the same manipulation on $\beta^+$. We conclude by induction on the intervals where $h_\beta\leq 0$. 

Thus, a solution of the relaxation has a virtual value positive everywhere and a reserve value equal to zero.  In this case, $U_r(\beta_i) = U(\beta_i)$.
\end{proof}
This new objective enables to run simple gradient descent algorithms without the need to recompute the reserve value at each iteration. It is also more stable than the original one since a local change of the virtual value does not completely change the value of the objective, which could be the case when the reserve value were part of the objective.
\section{Experimental setup}
We present in this section the complete approach and report the uplift of the new bidding strategies in various revenue-maximizing auctions.
\subsection{Our architecture}
\label{architecture}
To fit the optimal strategies, we use a simple one-layer neural network with $200$ ReLus.
We replace the indicator function by a sigmoid function to have  a fully differentiable  objective and we optimize
\begin{equation*}
U_\eta(\beta_i) = \mathbb{E}_{X_i \sim F_i}\bigg((X_i-h_{\beta_i}(X_i))G_i(\beta(X_i))\sigma(\eta h_{\beta_i}(X_i))\bigg)\;.
\end{equation*}
with $\sigma(x) = \frac{1}{1+\exp(-x)}$ and $\eta=1000$.
We start with a batch size of $10000$ examples, sampled according to the value distribution of the bidder. We use a stochastic gradient algorithm (SGD) with a decreasing learning rate starting at $0.001$. The full code in PyTorch is provided with the paper. The learning of an affine shading strategy is also provided in the notebook and is reaching already very decent performance.

In our setting, we assume that $G_i$ is  known. However, we could replace its expression by an approximation $\hat{G_i}$  learned from  past examples of bids of the competition or on the winning distribution of bidder $i$ computed on past auctions (in practice one may have to use survival analysis techniques to account for censoring of the observations). The results for the lazy second price auction with personalized price are presented in Table \ref{table_exponential} and in Table \ref{table_uniform}.
\subsection{Extension to other types of auction}
Our approach can easily be extended to many other types of auctions. Only a few lines of code are needed to adapt the objective to other mechanisms. 
\paragraph{The Myerson auction.} The Myerson auction \cite{Myerson81}  consists in using the virtual value both for the allocation  and  payment rules. The item is allocated to the bidder with the highest non-negative virtual value that pays:
\begin{equation*}
\psi_{B_i}^{-1}(\max(\max_{j\neq i}\psi_{B_j}(X_j),0))
\end{equation*}
As for the lazy second price auction, we can use \textit{the Myerson lemma} and show that the expected utility of the strategic bidder using the  strategy $\beta$ in the Myerson auction is
$$
U_i(\shadingFunc_i)=\Expb{[X_i-h_{\beta_i}(X_i)] F_Z(h_{\beta_i}(X_i))}\;.
$$
with $F_Z$ the cumulative distribution function of $Z=\max_{2\leq j \leq K}(0,\vValue_j(X_j))$, $X_i$ is the value of bidder $i$, and $h_{\beta_i} = \psi_{B_i}(\beta_i(X_i))$ is the virtual value function associated with the bid distribution.
For some distribution, the optimal strategy can be analytically computed. For instance, for the uniform distribution, we can prove this lemma which defines the optimal strategies. 
\begin{lemma}[Shading against $(K-1)$ uniform bidders]\label{lemma:shadingAgainstUnifBidders}
Suppose that $x$ has a positive density on its support and assume  that $x$ is bounded by $(K+1)/(K-1)$. Let $\epsilon>0$ be chosen by bidder 1 arbitrarily close to 0. 
Let us call 
$$
h^{(\eps)}_K(x)=
\begin{cases}
\frac{K-1}{K} \frac{\eps}{1+\eps} x & \text{ if } x \in [0,(1+\eps)/(K-1)) \;,\\
\frac{K-1}{K}\left(x-\frac{1}{K-1}\right) & \text{ if } x\geq (1+\eps)/(K-1)\;.
\end{cases}
$$
A near-optimal shading strategy is for bidder 1 to shade her value through
$$
\shadingFunc^{(\eps)}_1(x)=\Expb{h^{(\eps)}_K(t)|t\geq x}\;.
$$
As $\eps$ goes to $0^+$, this strategy approaches the optimum. 

If the support of $x$ is within $(1/(K-1),(K+1)/(K-1))$, then $\eps$ can be taken equal to 0.
\end{lemma}
The full proof is in Appendix \ref{MyersonOneStrategic}. Since in this specific setting optimal strategies have a known closed form, our optimization pipeline can be tested to see if it recovers these strategies. With the same pipeline used in Section \ref{architecture}, we optimize
$$
U_i(\shadingFunc_i)=\Expb{[X_i-h_{\beta_i}(X_i)] F_Z(h_{\beta_i}(X_i))}\;.
$$
Appendix \ref{comparisonMyerson} focuses on the uniform distribution where our algorithm recover exactly the strategies proposed in Lemma \ref{lemma:shadingAgainstUnifBidders} showing the robustness of our approach.

The interest of the optimization pipeline is the direct extension to all possible value distributions without the need to solve at each time a new system of differential equations. The performance with an exponential value distribution is provided in Table \ref{table_exponential}.

\begin{table*}[t]
\begin{center}
\begin{tabular}{|l|l|l|l|l|}
\hline
Auction Type                               &                                    & K=2 & K=3 & K=4 \\ \hline
\multirow{2}{*}{Baselines} & Utility of truthful strategy (in revenue maximizing)                                                                  & 0.30 &0.24  & 0.21
\\ \cline{2-5}
& Utility of truthful strategy (in welfare maximizing)                                                                 & 0.50 & 0.33  &  0.25
\\ \hlineB{4.0}
\multirow{2}{*}{Lazy second price auction} & Utility of strategic bidder        &   \tiny{$0.45 \pm 0.001$}  &  \tiny{$0.31 \pm 0.001$}   &   \tiny{$0.24 \pm 0.001$}  \\ \cline{2-5} 
                                           & Uplift vs truthful bidding       & +50\%    & +29\%    &   +14\%  \\ \hlineB{4.0}
\multirow{2}{*}{Eager second price auction}           & Utility of strategic bidder        &  \tiny{$0.52 \pm 0.02$}   &  \tiny{$0.33 \pm 0.02$}   & \tiny{$0.25 \pm 0.02$}      \\ \cline{2-5} 
                                           & Uplift vs truthful bidding      &   +73\%  & +37\%     &   +19\%  \\ \hlineB{4.0}
\multirow{2}{*}{Myerson auction}           & Utility of strategic bidder        & \tiny{$0.64\pm0.001$}    & \tiny{$0.45\pm0.001$}     & \tiny{$0.35\pm0.001$}     \\ \cline{2-5} 
                                           & Uplift vs truthful bidding        &  +113\%   &  +87\%   &   +67\%  \\ \hlineB{4.0}
\multirow{2}{*}{Boosted second price}      & Utility of strategic bidder        &  \tiny{$0.48\pm0.03$}   &  \tiny{$0.41\pm0.001$}   &  \tiny{$0.32\pm0.001$}   \\ \cline{2-5} 
                                           & Uplift vs truthful bidding        &  +60\%   &  +71\%   &   +52\%  \\  \hline
\end{tabular}
\caption{\textbf{All bidders have  an exponential value distribution with parameter $\lambda=1$. The strategic bidder has K-1 opponents bidding truthfully and having a reserve price equal to 1.0,  their monopoly price. The reserve price of the strategic bidder is computed on her bid distribution. For each run, the evaluation is based on $10^{6}$ samples, and we average the performances over 10 learnings.} The utility of the strategic bidder can be higher that in the welfare-maximizing auction because revenue maximizing auctions  remove competition below the reserve price, as illustrated by some examples in Appendix \ref{appendix_experiments}.}
\label{table_exponential}
\end{center}
\end{table*}

\paragraph{Eager second price auction with monopoly reserve prices.}
The eager second price auction consists of running a second price auction but only among bidders that clear their personalized reserve price. The objective function is very similar to the one of the lazy second price auction except that the winning distribution is different below the reserve price of the other bidders. Indeed, if all other bidders are below their reserve price, the strategic bidders that bids above his monopoly price is sure to win and only pays her monopoly price. We provide more details in Appendix \ref{eagerappendix}.

\paragraph{The boosted second price auction. (BSP)} 
Two small variants of the boosted second price auction (BSP) \cite{Golrezaei2017} can also be addressed.  We deal with the BSP auction as it seems to be one of the state of the art alternative to the second price auction with personalized reserve price to be used in practice and deals with heterogeneities between bidders. In the original paper, the seller computes first the reserve prices of each bidder based on their bid distributions. Then, the algorithm computes a boosting factor $\gamma_i>0$ for each bidder by counterfactually maximizing the revenue of the seller. More precisely, the auction is ran according to~:
\begin{algorithm}[h!]
\caption{Boosted second price $(r,\gamma)$}
\begin{algorithmic}
\STATE - First each bidder i submits his bid $b_i$
\STATE - Define S as a set of bidders whose bids exceed their reserve price, i.e, $S=\{i:b_i \geq r_i\}$
\STATE - If the set S is empty, the item is not allocated. Otherwise, the item is allocated to bidder $i^{*}$ with the highest boosted bid, i.e., $i*=\argmax_{i \in S}\{b_i\gamma_i\}$ and she pays $\max\{r_{i^{*}},\max_{i \in S, i \neq i^{*}}\{b_i\gamma_i/\gamma_{i^{*}}\}\}$. For other bidders, the payment is zero.
\end{algorithmic}
\end{algorithm}

To explain intuitively our two objectives corresponding to this auction, we consider first the example of the family of generalized Pareto distributions. As the virtual value of all distributions in this family is affine, the boosted second price auction is strictly equivalent to the Myerson auction in this family. It explains why this auction can perform well in practice for the seller since it avoids to compute exactly the virtual value by approximating it by a linear fit.  

In the first model, we assume that the seller first makes an affine-fit through an L2 regression on the virtual values she observes. Then, she runs a Myerson auction based on these L2 fits. In the case of Generalized Pareto distributions, this procedure results exactly in the BSP auction. If we note $\hat{\psi}_{B_i}$ the L2 fit of the virtual value corresponding to the bid distrution and $\hat{h}_{\beta_i} =  \hat{\psi}_{B_i}\circ \beta_i$, we optimize the Myerson objective with $\hat{h}_{\beta_i}$ corresponding to the fit of $h_{\beta_i}$. The main difference with the BSP auction for non-generalized pareto distributions is that the fit is used to compute the reserve price. In the second objective, we adress this limitation by computing first the reserve price $r_i$ based on the observed $\psi_{B_i}$. Then, the algorithm computes a linear fit of  $\psi_{B_i}$ on bids higher than $r_i$. This linear fit is used as the boosting parameter for bidder i. To make the objective differentiable, we consider the relaxation where $r_i$ is assumed to be $min_{r}(\psi_{B_i}(r) >0)$. We verify retrospectively that the final strategy verifies $min_{r}(\psi_{B_i}(r) > 0) = \argmax(r_i(1-F_{B_i}(r_i)))$

Our experiments show that our approach can also be empirically generalized to more advanced, intricate, practical and modern settings, on top of working well theoretically on the lazy second price auctions. 
\subsection{Evaluation and results}

Two different value distributions were used to run the experiments: the exponential distribution in Table \ref{table_exponential} and the uniform distribution in Table \ref{table_uniform} in Appendix. We focus on a small number of bidders since it is where the reserve price play an important role for the seller. \cite{celis2014buy} also noticed that the median of the number of participants in online advertising auctions is 6. 

To compute the real performance of the strategy, we are conservative in the computation of the reserve price since we use $r_i=\max(b | \psi_{B_i}(b)<0)$. We then compute the performance by computing the objective (expected utility) with Monte-Carlo simulations. For the lazy second price with personalized reserve price, we use for instance
\begin{equation*}
U(\beta_i) = \mathbb{E}\bigg((X_i-h_{\beta_i}(X_i))G_i(\beta(X_i))\textbf{1}(\beta_i(X_i) \geq r_i)\bigg)\;.
\end{equation*}
We compare the performance of our strategies with two baselines: the utility of one bidder bidding truthfully in a second price auction without reserve price (the welfare maximizing auction) and in a second price auction with monopoly price (with symmetric bidders, this auction is equivalent to the Myerson auction and is revenue-maximizing for the seller).

For BSP, we report results for the second objective which is the closest one to the corresponding procedure of \cite{Golrezaei2017}. The first one gives similar uplifts that the strategic behavior in the Myerson auction.  The order of magnitude of the uplift reported is  significant. We observe that BSP and the Myerson auction are less robust to strategic behavior than the lazy second price auction with personalized reserve price. Indeed, as in the eager version of the second price auction there is no competition when all other bidders are below their reserve price. It is not the case for the lazy second price auction explaining why the uplift are slightly lower for this specific auction. 

We focused  on the stationary case where the strategic bidder has to choose one strategy implying a bid distribution and the seller will immediately optimize their mechanism according to this bid distribution. However, our differential approach allows some generalizations. In future work, we could adapt the differential approach to more dynamic settings where the seller uses a particular dynamic to update the reserve price based on past bids of the bidders. 
\section{Conclusion} 
In this paper, we showed that machine learning can be efficiently used on the bidder side to learn how to shade in revenue-maximizing auctions that are optimized based on past bids (or a distribution announced by the bidder to which she commits). Our work, both theoretical and practical, complements the classical approach using statistical learning from the seller's standpoint showing that strategic bidding can be implemented in some of the main revenue-maximizing auctions. Our work also raises questions about many automatic mechanism procedures since many are based on the assumption of having observed past truthful bids in order to optimize mechanisms. From an industry point of view, our work provides a new argument to come back to simple and more transparent auction mechanisms that are less subject to optimization on both the bidders' and the seller's sides. 
\newpage
\section*{Acknowledgements}
V. Perchet has benefited from the support of the FMJH Program Gaspard Monge in optimization and operations research (supported in part by EDF) and from the CNRS through the PEPS program. N. El Karoui gratefully acknowledges support from grant NSF DMS 1510172.
\bibliographystyle{ACM-Reference-Format}
\bibliography{gradient_final}

% This document was modified from the file originally made available by
% Pat Langley and Andrea Danyluk for ICML-2K. This version was created
% by Iain Murray in 2018, and modified by Alexandre Bouchard in
% 2019. Previous contributors include Dan Roy, Lise Getoor and Tobias
% Scheffer, which was slightly modified from the 2010 version by
% Thorsten Joachims & Johannes Fuernkranz, slightly modified from the
% 2009 version by Kiri Wagstaff and Sam Roweis's 2008 version, which is
% slightly modified from Prasad Tadepalli's 2007 version which is a
% lightly changed version of the previous year's version by Andrew
% Moore, which was in turn edited from those of Kristian Kersting and
% Codrina Lauth. Alex Smola contributed to the algorithmic style files.
\appendix
\newpage
\onecolumn
\section{Results for the uniform distribution}
\begin{table*}[h!]
\small
\centering
\begin{tabular}{|l|l|l|l|l|}
\hline
Auction Type                               &                                    & K=2 & K=3 & K=4 \\ \hline
\multirow{2}{*}{Baselines} & Utility of truthful strategy (in revenue maximizing)                                                                  & 0.083 & 0.057 & 0.040 
\\ \cline{2-5}
& Utility of truthful strategy (in welfare maximizing)                                                                 & 0.166 & 0.083 & 0.050 
\\ \hlineB{4.0}
\multirow{2}{*}{Lazy second price auction} & Utility of strategic bidder        &  \tiny{$0.141\pm0.001$}    &  \tiny{$0.077\pm0.001$}   & \tiny{$0.048\pm0.001$}    \\ \cline{2-5} 
                                           & Uplift vs truthful bidding      &   +72\%  &   +36\%  &  +20\%   \\ \hlineB{4.0}
\multirow{2}{*}{Eager second price auction}           & Utility of strategic bidder        & \tiny{$0.126 \pm 0.01$}     &  \tiny{$0.083 \pm 0.01$}   &    \tiny{$0.050 \pm 0.002$}   \\ \cline{2-5} 
                                           & Uplift vs revenue-maximizing       & +51\%    & +46\%   & +25\%    \\ \hlineB{4.0}
\multirow{2}{*}{Myerson auction}           & Utility of strategic bidder        & \tiny{$0.246\pm0.001$}   &   \tiny{$0.131\pm0.01$}  &   \tiny{$0.079\pm0.001$}   \\ \cline{2-5} 
                                           & Uplift vs revenue-maximizing       &  +195\%   &    +130\% &   +97.5\%  \\ \hlineB{4.0}
\multirow{2}{*}{Boosted second price}      & Utility of strategic bidder        &    \tiny{$0.24 \pm 0.01$}&   \tiny{$0.08 \pm 0.01$}  &  \tiny{$0.055 \pm 0.002$}  \\ \cline{2-5} 
                                           & Uplift vs revenue maximizing       & +200\%    &  +40\%  &  +37.5\% \\ \hline
\end{tabular}
\caption{\textbf{All bidders have a uniform value distribution. The strategic bidder has $K-1$ opponents, all bidding truthfully. The reserve price of all other bidders is equal to 0.5. The reserve price of the strategic bidder is computed on her bid distribution. For each run, the evaluation is based on $10^{6}$ samples. We average on 10 learnings the performance of the strategies.} The utility of the strategic bidder can be higher that in the welfare-maximizing auction because revenue maximizing auctions are removing the competition below the reserve price. We provide some examples of strategies in Appendix \ref{appendix_experiments}.}
\label{table_uniform}
\end{table*}

\input{formalProofKKT}

\section{Proof of the results on Myerson auction with one strategic bidder}
\label{MyersonOneStrategic}
The Myerson auction \cite{Myerson81}  consists in using the virtual value to both define the allocation rule and the payment rule. The item is allocated to the bidder with the highest non-negative virtual value and she pays:
\begin{equation*}
\psi_{B_i}^{-1}(\max(\max_{j\neq i}\psi_{B_j}(X_j),0))
\end{equation*}
As for the lazy second price auction, we can use \textit{the Myerson lemma} and show that the expected utility of the strategic bidder using the bidding strategy $\beta$ in the Myerson auction is
$$
U_i(\shadingFunc_i)=\Expb{[X_i-h_{\beta_i}(X_i)] F_Z(h_{\beta_i}(X_i))}\;.
$$
with $F_Z$ the cumulative distribution function of $Z=\max_{2\leq j \leq K}(0,\vValue_j(X_j))$, $X_i$ is the value of bidder $i$, and $h_{\beta_i} = \psi_{B_i}(\beta_i(X_i))$ is the virtual value function associated with the bid distribution.
Suppose that $\shadingFunc \mapsto \shadingFunc_t=\shadingFunc +  t \rho$, where $t>0$ is small and $\rho$ is a function. We note that $h_{\beta+t\rho}(x)=h_\beta+t h_\rho$. We have the following result.

\begin{lemma}\label{lemma:directDerivativeMyersonAuction}
Suppose we change $\beta$ into $\beta_t=\beta+t \rho$. Both $\beta$ and $\beta_t$  are assumed to be non-decreasing. Call $x_{\beta}$ the reserve value corresponding to $\beta$, assume it has the property that $h_\beta(x_\beta) = 0$ and $h_\beta'(x_\beta)\neq 0$ ($h_\beta'$ is assumed to exist locally). Assume $x_\beta$ is the unique global maximizer of the revenue of the seller. Then, 
\begin{align*}
\left.\frac{\partial}{\partial t} U(\beta_t) \right|_{t=0} 
&=\Expb{h_\rho(X) [(X-h_{\beta}(X))f_Z(h_{\beta}(X))-F_Z(h_{\beta}(X))]\indicator{X>x_{\beta}}}
\\
&+\frac{h_\rho(x_{\shadingFunc})}{h'_{\shadingFunc}(x_{\shadingFunc})}\prod_{i=2}^K F_{V_i}(0) f_{1}(x_{\shadingFunc}) x_{\shadingFunc}\;,\notag
\end{align*}
\end{lemma}
\proof{Proof.}
Taking directional derivative of the utility of the bidder gives the equation.
\endproof
In the work below, we naturally seek a shading function $\shadingFunc$ such that these directional derivatives are equal to 0. We will therefore be interested in particular in functions $\shadingFunc$ such that $[x-\vValue_B(\shadingFunc(x))]f_Z(\vValue_B(\shadingFunc(x)))=F_Z(\vValue_B(\shadingFunc(x)))$, when $\vValue_B(\shadingFunc(x))>0$. The second term in our equation has intuitively to do with the event where the other bidders are discarded for not beating their reserve price. As we will see below, we can sometimes ignore this term, for instance when an equilibrium strategy exists which amounts to canceling the reserve value ($x_\beta = 0$ in this case). 

We will be keenly interested in shading functions $\shadingFunc$ such that $$
(x-\vValue_B(\shadingFunc(x)))f_Z(\vValue_B(\shadingFunc(x)))=F_Z(\vValue_B(\shadingFunc(x)))\;, \text{ when } \vValue_B(\shadingFunc(x))>0\;. 
$$
Indeed, for those $\shadingFunc$'s, the expectation in our differential will be 0. Hence, computing the differential will be relatively simple and in particular will give us reasonable guesses for $\shadingFunc$ and descent directions, even if it does not always give us directly an optimal shading strategy. Furthermore, when $K$ is large, the second term fades out, as the probability that no other bidder clear their reserve prices becomes very small. 

If we proceed formally, and call, for $x>0$, $h(x)=(\text{id}+F_Z/f_Z)^{-1}(x)$ (temporarily assuming that this - possibly generalized - functional inverse can be made sense of), we see that solving the previous equation amounts to solving
$$
\vValue_B(\shadingFunc(x))=(\text{id}+F_Z/f_Z)^{-1}(x)=h(x)\;.
$$
Lemma \ref{definition_psi} can of course be brought to bear on this problem. We note that we will be primarily interested in solutions of this equation for $x$'s such that $\vValue_B(\shadingFunc(x))>0$. 

\subsection{Explicit computations in the case of Generalized Pareto families}
It is clear that now we need to understand $F_Z$, $F_Z/f_Z$ and related quantities to make progress. By definition, if $X$ is Generalized Pareto (GP) with parameters $(\mu,\sigma,\xi)$, we have, when $\xi<0$ and $t\in [\mu,\mu-\sigma/\xi]$, 
$$
P(X\geq t)=(1+\xi (t-\mu)/\sigma)^{-1/\xi}\;
$$
and otherwise $P(X\geq t)=\Expb(-(t-\mu)/\sigma)$ if $\xi=0$.  In GP families, the virtual value has the form $\vValue(t)=c_\vValue(t-r^*)$, where $r^*$ is the monopoly price and $c_\vValue=1-\xi$.

\begin{lemma}\label{lemma:variousGPComps}
Suppose $Y$ has a Generalized Pareto distribution. Call $F_Y$ the cdf of $Y$ and $f_Y$ its density. 

If $V=\vValue_Y(Y)$, where $\vValue_Y$ is the virtual value of $Y$, we have $F_V(t)=F_Y(\vValue_Y^{-1}(t))$ and 
$$
\frac{F_V(t)}{f_V(t)}=c_\vValue \frac{F_Y(\vValue_Y^{-1}(t))}{f_Y(\vValue_Y^{-1}(t))}\;,
$$
where $c_{\vValue_Y}=\vValue_Y'(t)$. If $r^*_Y$ is the monopoly price associated with $Y$, we more specifically have 
$$
\vValue_Y^{-1}(t) = \frac{t}{c_{\vValue_Y}}+r^*_Y\;, \text{ and }\frac{F_V(t)}{f_V(t)}=c_{\vValue_Y} \frac{F_Y(t/c_{\vValue_Y}+r^*_Y)}{f_Y(t/c_{\vValue_Y}+r^*_Y)}\;.
$$
\end{lemma}
In case $F_Y$ has finite support, we naturally restrict $t$ to values such that $x=\vValue_Y^{-1}(t)$ is just that $f_Y(x)>0$. 
\begin{proof}
$F_V$ is just the cumulative distribution function of $\psi_Y(Y)$, where $\vValue_Y$ is the virtual value of $Y$. Hence, since in GP families $\vValue_Y$ is increasing, 
$$
F_V(t)=P(V\leq t)=P(\vValue_Y(Y)\leq t)=F_Y(\vValue_Y^{-1}(t))\;.
$$
In particular, 
$$
f_V(t)=\frac{f_Y(\vValue_Y^{-1}(t))}{\vValue_Y'(\vValue_Y^{-1}(t))}\;.
$$
In Generalized Pareto families, $\vValue_Y$ is linear, so that $\vValue_Y'$ is a constant, because $\vValue_Y(t)=c_{\vValue_Y}(t-r^*_Y)$, where $r^*_Y$ is the monopoly price. The first result follows immediately. Noticing that $\vValue_Y^{-1}(x)=x/c_{\vValue_Y}+r^*_Y$ gives the second result.  
\end{proof}
The previous lemma yields the following useful corollary. 
\begin{corollary}\label{coro:keyManipsMyersonGPopponents}
Suppose $K\geq 2$, $Y_2,\ldots,Y_K$ are independent, identically distributed, with Generalized Pareto distribution. Call $\vValue_Y$ their virtual value function and $Z=\max_{2\leq i \leq K}(0,\vValue_Y(Y_i))$. Then, if $F_Z$ is the cumulative distribution function of $Z$, we have 
$$
\frac{F_Z(t)}{f_Z(t)}=\frac{c_{\vValue_Y}}{K-1} \frac{F_Y(t/c_{\vValue_Y}+r^*_Y)}{f_Y(t/c_{\vValue_Y}+r^*_Y)}\;, \text{ for } t>0 \text{ and such that } f_Y(t/c_{\vValue_Y}+r^*_Y)>0\;.
$$
\end{corollary}
\subsection{An example: uniform non-strategic bidders}\label{app:subsec:myersonOneStratUnif}
In this subsection we assume that bidder 1 is facing $K-1$ other bidders, with values $Y_i$'s that are i.i.d $\Unif[0,1]$. In this case, $c_{\vValue_Y}=2$ and $r^*_Y=1/2$ so $F_Z(t)=\min(1,[(t+1)/2]^{K-1})$ for $t>0$.   
We recall that the $\Unif[0,1]$ distribution is GP(0,1,-1). Bidder 1 is strategic whereas bidders 2 to $K$ are not and bid truthfully. 

\begin{lemma}[Shading against $(K-1)$ uniform bidders]\label{lemma:shadingAgainstUnifBidders}
Suppose that $x$ has a density that is positive on its support. We assume for simplicity that $x$ is bounded by $(K+1)/(K-1)$. Let $\epsilon>0$ be chosen by bidder 1 arbitrarily close to 0. 
Let us call 
$$
h^{(\eps)}_K(x)=
\begin{cases}
\frac{K-1}{K} \frac{\eps}{1+\eps} x & \text{ if } x \in [0,(1+\eps)/(K-1)) \;,\\
\frac{K-1}{K}\left(x-\frac{1}{K-1}\right) & \text{ if } x\geq (1+\eps)/(K-1)\;.
\end{cases}
$$
A near-optimal shading strategy is for bidder 1 to shade her value through
$$
\shadingFunc^{(\eps)}_1(x)=\Expb{h^{(\eps)}_K(t)|t\geq x}\;.
$$
As $\eps$ goes to $0^+$, this strategy approaches the optimum. 

If the support of $x$ is within $(1/(K-1),(K+1)/(K-1))$, then $\eps$ can be taken equal to 0.
\end{lemma}

\begin{proof}
If we call $h(x)=\psi_B(\shadingFunc(x))$ we can in this case write bidder 1's expected payoff directly using the results of the previous subsection: 
$$
U(\shadingFunc)=\int_{x:h(x)>0} (x-h(x))\min\left(1,\frac{[h(x)+1]^{K-1}}{2^{K-1}}\right) f_1(x) dx\;.
$$
In light of the fact we want to maximize this integral as a function of $h$, with the requirement that $h> 0$, it is natural to study the function $f_x(c)=(x-c)[c+1]^{K-1}$. 

If we call $h_K(x)=\argmax_{c\geq 0} f_x(c)$, we can split the problem into two cases. If $x>1/(K-1)$, $h_K(x)=\frac{K-1}{K}\left(x-\frac{1}{K-1}\right)$. Note that with our assumption that $x\leq (K+1)/(K-1)$, $h_K(x)\leq 1$. For $x<1/(K-1)$, the function $f_x(\cdot)$ is decreasing for $c\geq 0$. Hence, 
$$
h_K(x)=\argmax_{c\geq 0}(x-c)[c+1]^{K-1}=
\begin{cases}
0 & \text{ if } x\leq 1/(K-1)\;,\\
\frac{K-1}{K}\left(x-\frac{1}{K-1}\right) & \text{ if } x>1/(K-1) \;.
\end{cases}
$$
Recall that for Lemma \ref{definition_psi} to apply, we need to integrate an increasing function and $h_K$ is not increasing on $(0,\infty)$. 

However, bidder 1 can use the following $\eps$-approximation strategy: let us call 
$$
h^{(\eps)}_K(x)=
\begin{cases}
\frac{K-1}{K} \frac{\eps}{1+\eps} x & \text{ if } x \in [0,(1+\eps)/(K-1)) \;,\\
\frac{K-1}{K}\left(x-\frac{1}{K-1}\right) & \text{ if } x\geq (1+\eps)/(K-1)\;.
\end{cases}
$$
Notice that $\sup_x |h^{(\eps)}_K(x)-h_K(x)|<\eps/K$. In light of Lemmas \ref{definition_psi} and \ref{lemma:directDerivativeMyersonAuction}, the corresponding  function
$$
\shadingFunc^{(\eps)}_1(x)=\Expb{h^{(\eps)}_K(x)|x\geq x}
$$
is increasing and will then guarantee an expected payoff that is nearly optimal since it will be 
$$
U(\shadingFunc^{(\eps)})=\frac{1}{2^{K-1}}\int (x-h_K^{(\eps)}(x))[h_K^{\eps)}(x)+1]^{K-1} f_1(x) dx\;.
$$
It can be made arbitrarily close to optimal by decreasing $\eps$. Using $\eps>0$ guarantees that the virtualized bid $h^{(\eps)}_K(x)$ is always strictly positive and hence effectively sends the monopoly price for bidder 1 to 0. (Even if $\shadingFunc^{(0)}$ is increasing, a potential problem might occur if the virtualized bid is exactly zero. Using $\shadingFunc^{(\eps)}$ with $\eps=0^+$ solves that problem. We can also verify a posteriori that it also avoids potential problems related to ironing, which justifies post-hoc the formulation of the utility we used in the first place.)

If $x$ is supported on a subset of $[1/(K-1),(K+1)/(K-1))$, taking $\eps=0$ is possible and optimal. 
\end{proof}
The assumption that $x\leq (K+1)/(K-1)$ can easily be dispensed of as the proof makes clear : one simply needs to look for the argmax of another function. Our main example follows and does not require taking care of this minor technical problem. 

\paragraph{Case where bidder 1 has value distribution $\mathcal{U}[0,1]$} We first note that $x\leq 1\leq (K+1)/(K-1)$, so Lemma \ref{lemma:shadingAgainstUnifBidders} applies as-is. We therefore have 
$$
\shadingFunc^{(\eps)}_1(x)=
\begin{cases}
\frac{K-1}{K}[\frac{1}{2}(1+x)-\frac{1}{K-1}] & \text{ if } 1\geq x\geq x_{\eps}=\frac{1+\eps}{K-1}\;,\\
\frac{K-1}{K} \frac{1}{1-x}\left(\frac{\eps}{1+\eps} \frac{1}{2} (x_\eps^2-x^2)+\shadingFunc^{(\eps)}_1(x_\eps) (1-x_\eps)\right) & \text{ if } x<\frac{1+\eps}{K-1}\;.
\end{cases}
$$
Taking $\eps$ to 0 yields
$$
\shadingFunc_1(x)=
\begin{cases}
\frac{K-1}{K}[\frac{1}{2}(1+x)-\frac{1}{K-1}] & \text{ if } x\geq \frac{1}{K-1}\;,\\
\frac{1}{1-x} \frac{(K-2)^2}{2(K-1)K} & \text{ if } x<\frac{1}{K-1}\;.
\end{cases}
$$
See Figure \ref{fig:uniform_bid_profiles} for a plot of $\shadingFunc_1$ and comparison to other possible shading strategies.  

Similar computations can be carried out if $x$ has another GP distribution. For those distributions, the shading beyond $1/(K-1)$ is also affine in the value of bidder 1, $x$. Interestingly, it is easy to verify that affine transformations of GP random variables are GP. However, if the support of $x$ includes part of $(0,1/(K-1))$, $\shadingFunc_1(x)$ will not have a GP distribution in general.

\begin{figure}
\begin{center}
	\includegraphics[width=0.8\textwidth]{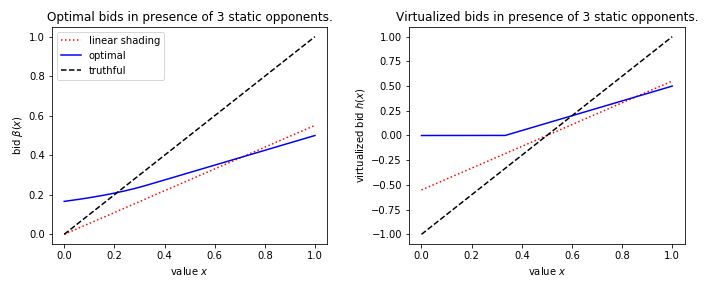}
\end{center}
	\caption{\textbf{Myerson auction: Bids and virtualized bids with one strategic bidder} There are K=4 bidders, only one of them is strategic. On the left hand side, we present a plot of the bids sent to the seller. ``Linear shading" corresponding to a bid $\shadingFunc_\alpha(x)=\alpha x$, where $x$ is the value of bidder 1; here $\alpha$ is chosen numerically to maximize that buyer's payoff - see Lemma \ref{lemma:directDerivativeMyersonAuction}. ``Optimal" corresponds to the strategy described in Lemma \ref{lemma:shadingAgainstUnifBidders}, with $\eps=0^+$. On the right hand side (RHS), we present the virtualized bids, i.e. the value taken by the associated virtual value functions evaluated at the bids sent to the seller. This corresponds on average to what the buyer is paying in the Myerson auction, when those virtualized bids clear 0. We can interpret the RHS figure as showing that for both optimal and linear shading,  a strategic buyer end up winning more often and paying less (conditional on the fact that she won) than if she had been truthful; this explains why her average payoff is higher than with truthful bidding.}
\label{fig:uniform_bid_profiles}	
\end{figure}

\subsection{Further computations in the Generalized Pareto case} 
\begin{lemma}\label{lemma:OneStrategicMyersonInGP}
Suppose a strategic bidder faces $K-1$ opponents sharing the same distribution $F_Y$ in the Generalized Pareto family. Then, assuming that the seller is welfare benevolent, her optimal shading function is such that her virtualized bid $h_{\text{optimal}}(x)$ satisfies 
$$
h_{\text{optimal}}(x)=\max(0,\psi_Y(\beta^I_{Y,K}(\psi_Y^{-1}(x))))\;.
$$
where $\beta^I_{Y,K}$ is the first price bid of a bidder facing competition with cdf $G=F_Y^{K-1}$.
The corresponding shading function $\beta$ can be easily obtained by an application of Lemmas \ref{definition_psi} and \ref{lemma:directDerivativeMyersonAuction}.
	
\end{lemma}
\textbf{Comment~:} we can of course find $h_{\text{optimal}}^{\eps}$ as above to approximate $h_{\text{optimal}}$ by an increasing positive function that is arbitrarily close to our target and avoid various technical issues. 
\begin{proof}{\textbf{Proof of Lemma \ref{lemma:OneStrategicMyersonInGP}}}
The computation for independent uniform opponents generalizes easily to opponents with value distributions in other GP families. In particular we need to find $h$ that maximizes
$$
\int_{x:h(x)\geq 0}(x-h(x))\bigg[F_Y\left(\frac{h(x)}{c_{\psi_Y}}+r_Y^*\right)\bigg]^{K-1} f_1(x) dx\;.
$$
We can maximize point by point and hence we are looking for $t^*(x)$ such that
$$
t^*(x)=\argmax_t(x-t)F_Y\left(\frac{t}{c_{\psi_Y}}+r_Y^*\right)^{K-1}\;, t>0\;.
$$ 
Differentiating the above expression gives 
$$
\delta(t)=f_Y(t/c_{\psi_Y}+r)F_Y^{K-2}(t/c_{\psi_Y}+r^*_Y)\left[\frac{x-t}{c_{\psi_Y}}-\frac{1}{K-1}
\frac{F_Y}{f_Y}(t/c_{\psi_Y}+r^*_Y)\right]\;.
$$
The expression in the bracket can be written $\psi_Y^{-1}(x)-H(\psi_Y^{-1}(t))$ where $H=\textrm{id}+\frac{G_{Y,K-1}}{g_{Y,K-1}}$, where $G_{Y,K-1}$ is the cdf of the max of $K-1$ i.i.d random variables and $g_{Y,K-1}$ its derivative. Elementary computations show that this function is increasing in GP families. In fact its derivative can be shown to be $1+(1-\frac{\xi}{\sigma}G_{}(x)/(1-G(x)))/(K-1)$ and $\xi<0$. Hence $H(\psi_Y^{-1}(t))$ is also increasing. Hence $\delta(t)$ is a decreasing function of $t$. It is also trivially continuous in GP families. We conclude that the equation $\delta(t)=0$ has at most 1 positive root. 

If $\psi_Y^{-1}(x)<H(\psi^{-1}_Y(0))$, we see that $\delta(t)<0$ for $t\geq 0$, in which case $t^*=0$. 
If that is not the case, then $\psi^{-1}_Y(t^*)=H^{-1}(\psi_Y^{-1}(x))$. Hence we have shown that 
$t^*=\max(0,\psi_Y(H^{-1}\psi_Y^{-1}(x)))$. Now we notice that the $H^{-1}(x)$ is nothing but the first price bid of a bidder facing competition with cdf $G=F_Y^{K-1}$, a bid function we denote by $\beta^{I}_{Y,K}$. So we conclude that 
$$
h_{\text{optimal}}(x)=\max(0,\psi_Y(\beta^I_{Y,K}(\psi_Y^{-1}(x))))\;.
$$
Once again the fact that $h_{\text{optimal}}$ is non-decreasing (as a composition of non-decreasing functions) avoids issues related to ironing.
\end{proof}
\newpage
\section{Notes on learning with splines}
\label{appendix_splines}
A natural question is  whether the buyer can compute shading strategies  numerically. A first approach is to look back at the gradient of the bidder's utility in the direction of a certain function $\rho$, i.e.,  the directional derivative, that can be computed by elementary calculus
\begin{align*}
\left.\frac{\partial}{\partial t} U_i(\beta_t)\right|_{t=0}&=\Expb{\left\{g(\beta(X_i))[X_i-\beta(X_i)]\right\}\rho(X_i)\indicator{X_i\geq x_\beta}}\\
&+G(\beta(x_{\beta}))\left[\frac{x_{\beta}f(x_{\beta})}{h_{\beta}^{'}(x_\beta)}-(1-F(x_{\beta}))\right]\rho(x_{\beta})\\
&-\frac{\rho'(x_{\beta})x_{\beta}(1-F(x_{\beta}))G(\beta(x_{\beta}))}{h_{\beta}^{'}(x_\beta)}\;.
\end{align*}
and to look at shading function expressed in a specific basis as$$
\beta(x)=\sum_{i=1}^N c_i(\beta) f_i(x)\;,
$$
and try to optimize over $c_i$. It would be also  quite natural to do  an isotonic regression and optimize over non-decreasing functions directly; this approach is tackled later on. 

\paragraph{A natural basis} Splines (see e.g. \cite{htfElementsStatLearning} for a practical introduction) are a natural candidate for the function $f_i$'s. In particular, first order splines are piecewise continuous functions, hence evaluating derivatives is trivial and it is easy to account in the formula above for the finitely many discontinuities of the derivative that will arise. If $\xi_k$'s are given knots, first order splines are the functions $$
f_1=1, f_2(x)=x, f_{k+2}(x)=(x-\xi_k)_+=\max(x-\xi_k,0)\;.
$$
Higher order splines could of course also be used. 
\begin{lemma}
As described above, the optimal shading problem can be numerically approximated using steepest descent by a succession of linear programs, provided the non-decreasing constraint on $\beta$ can be written linearly in $c_i$. This is of course the case for 1st order spline.
\end{lemma}
\begin{proof}
After the function is expanded in a basis, the functional gradient  becomes a standard gradient, and the shading function can be improved with a steepest descent.  If the reserve value is not one of the knots, the gradient above is easy to compute:  each step of the optimization requires to solve a constrained LP to ensure that the solution is increasing. 

For 1st order splines, because the derivative is constant between knots, checking that $\beta'(x)\geq 0$ for all $x$ amounts to checking finitely many linear constraints and hence is amenable to an LP. 
\end{proof}

This approach amounts to an exploration of parts of the space of shading functions. Since the objective is not even continuous, though differentiable in a large part of the parameter space, the optimization problem is hard. In our experiments, we got significant improvement over bidding truthfully by using the above numerical method. However, we encountered the discontinuities of the optimization problem described above: our numerical optimizer got stuck at shading functions around which the reserve value was very unstable, which corresponds to revenue curves for the seller with several distant (approximate) local maxima: a small perturbation in function space does not induce much loss of the revenue on the seller side, but can have a huge impact on the reserve value and hence the buyer revenue.

\subsection{Some numerical experiments with splines.}
\begin{figure}[h!]
\begin{tabular}{ccc}
\includegraphics[width=0.33\textwidth, height=0.19\textheight]{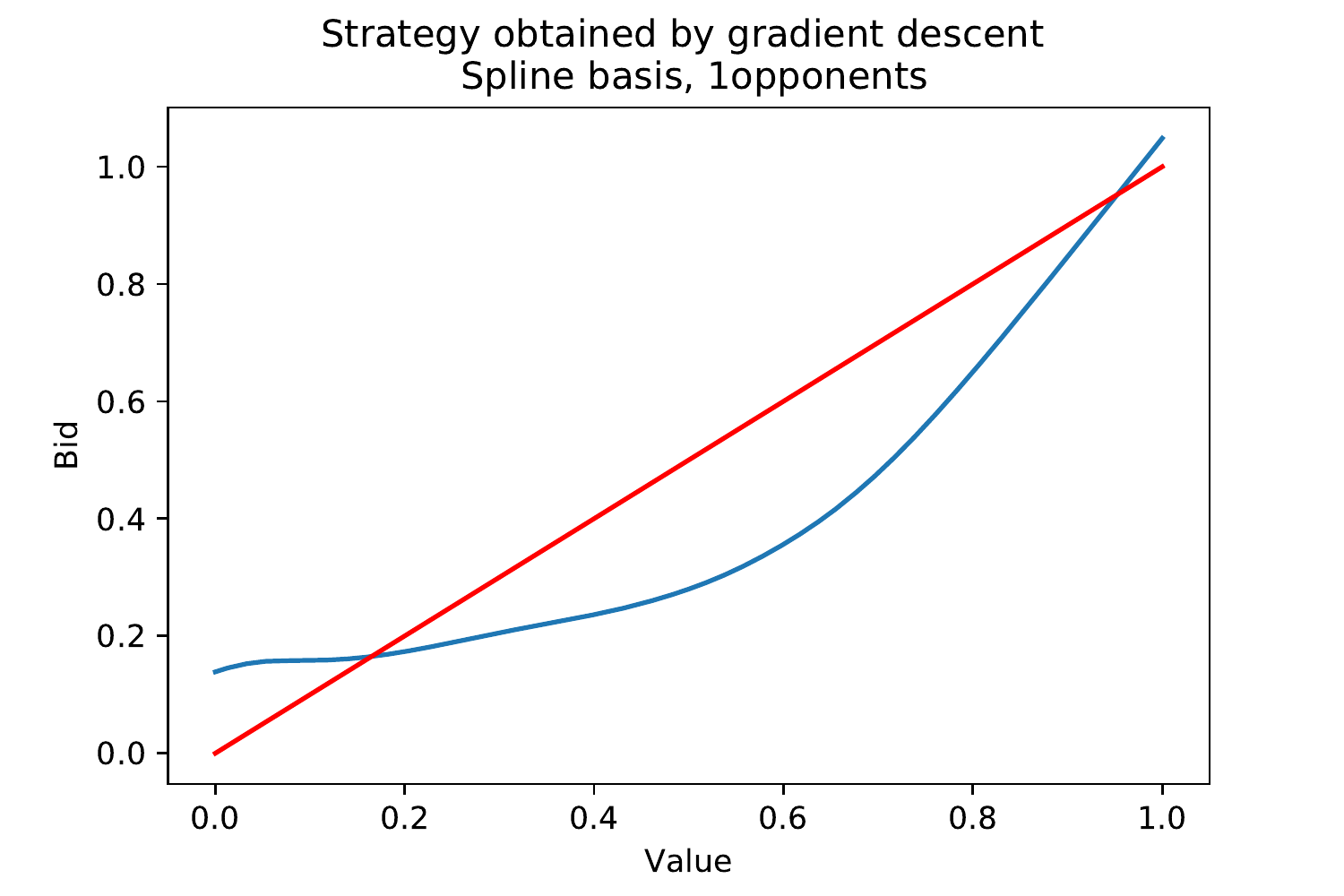} &
\includegraphics[width=0.33\textwidth, height=0.19\textheight]{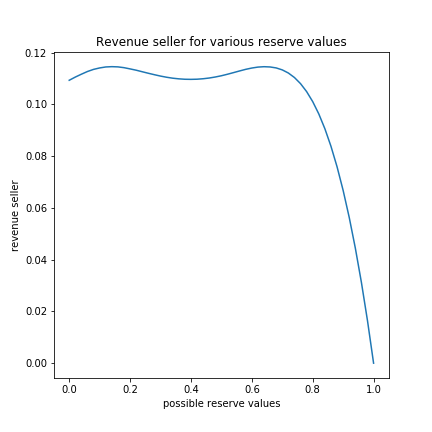}&
\includegraphics[width=0.33\textwidth, height=0.19\textheight]{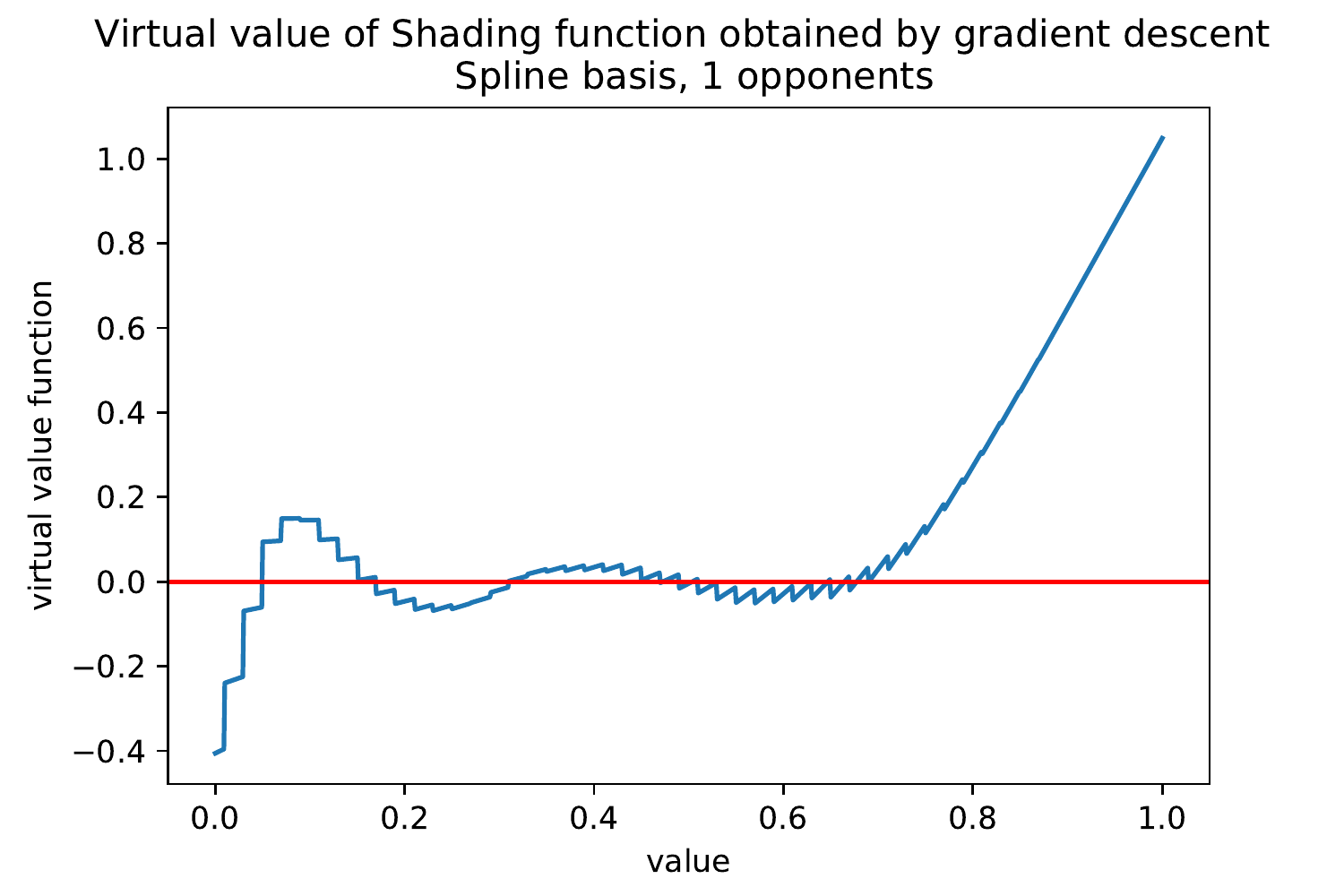}
 \end{tabular}
\caption{\textbf{At the end of gradient descent. Left: Strategy mapping bids to values. Middle: Seller revenue as a function of the reserve value. Left: Virtual value corresponding of the bid distribution as a function of the value.} The reserve value is equal to 0.21 and the reserve price to 0.2. We are in the case of two bidders with uniform distribution. One is strategic and the other one is bidding truthfully.}
\label{fig:end_grad_descent}
\end{figure}

\begin{figure}[h!]
\begin{tabular}{cccc}
\includegraphics[width=0.33\textwidth, height=0.19\textheight]{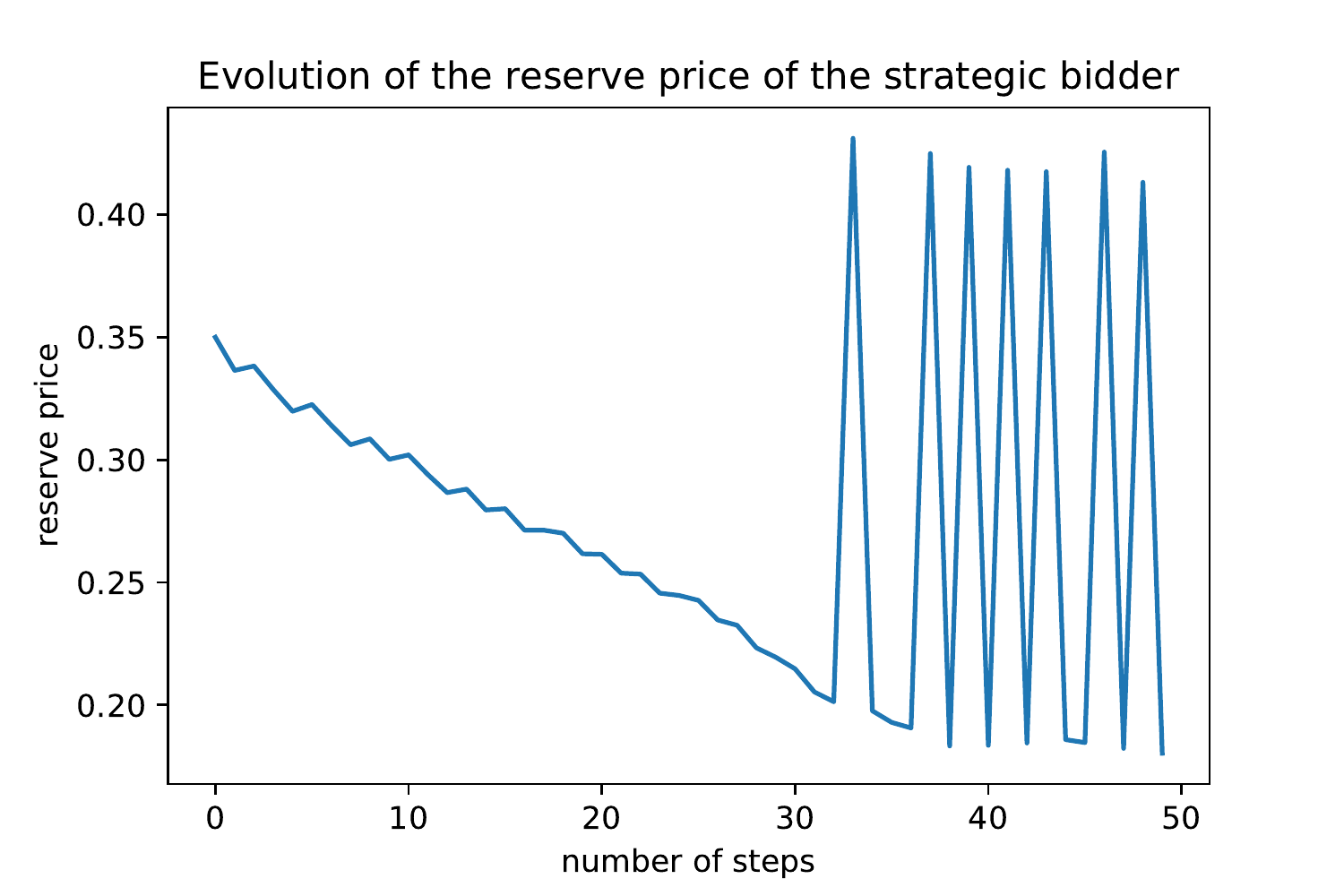} &
\includegraphics[width=0.33\textwidth, height=0.19\textheight]{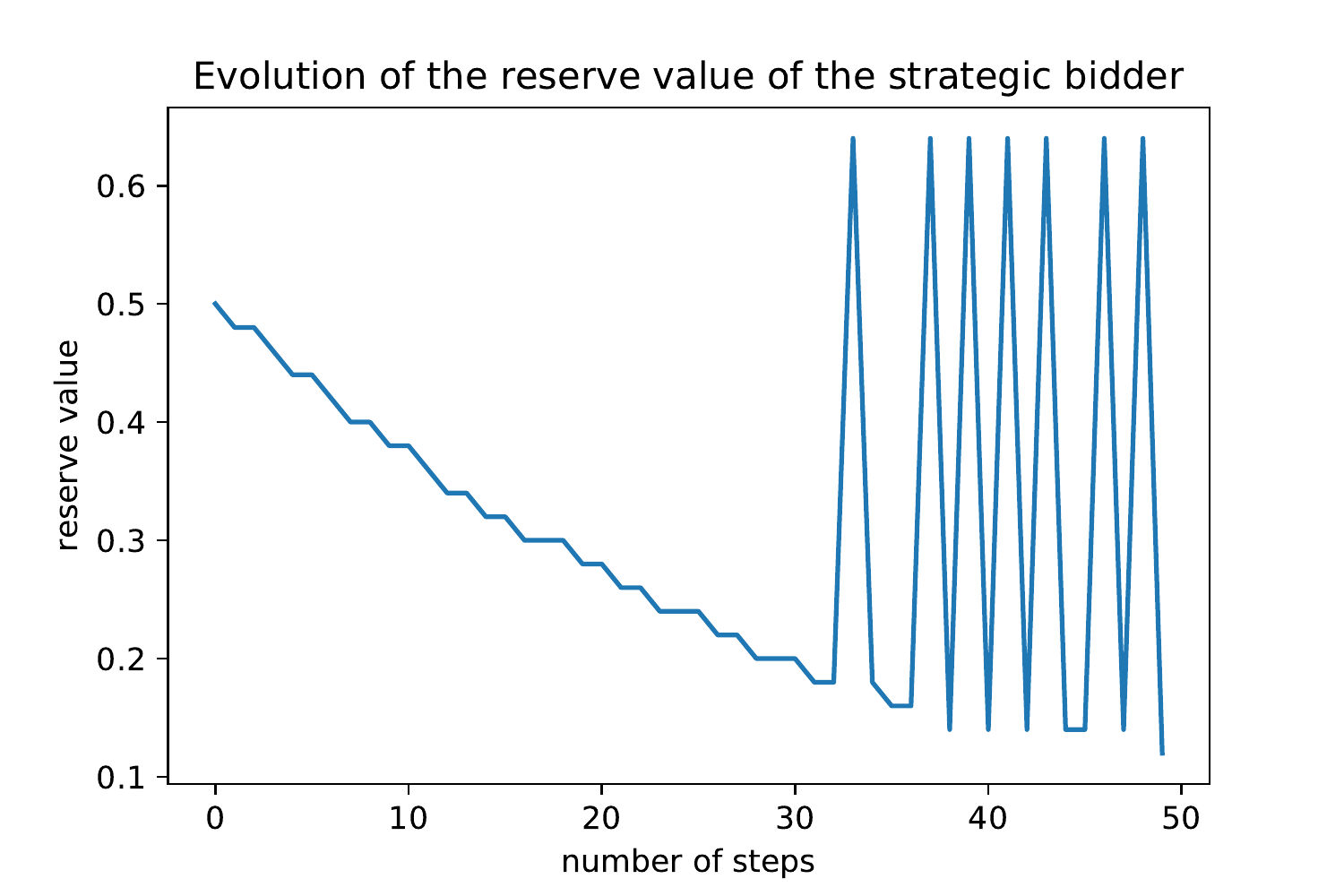}&
\includegraphics[width=0.33\textwidth, height=0.19\textheight]{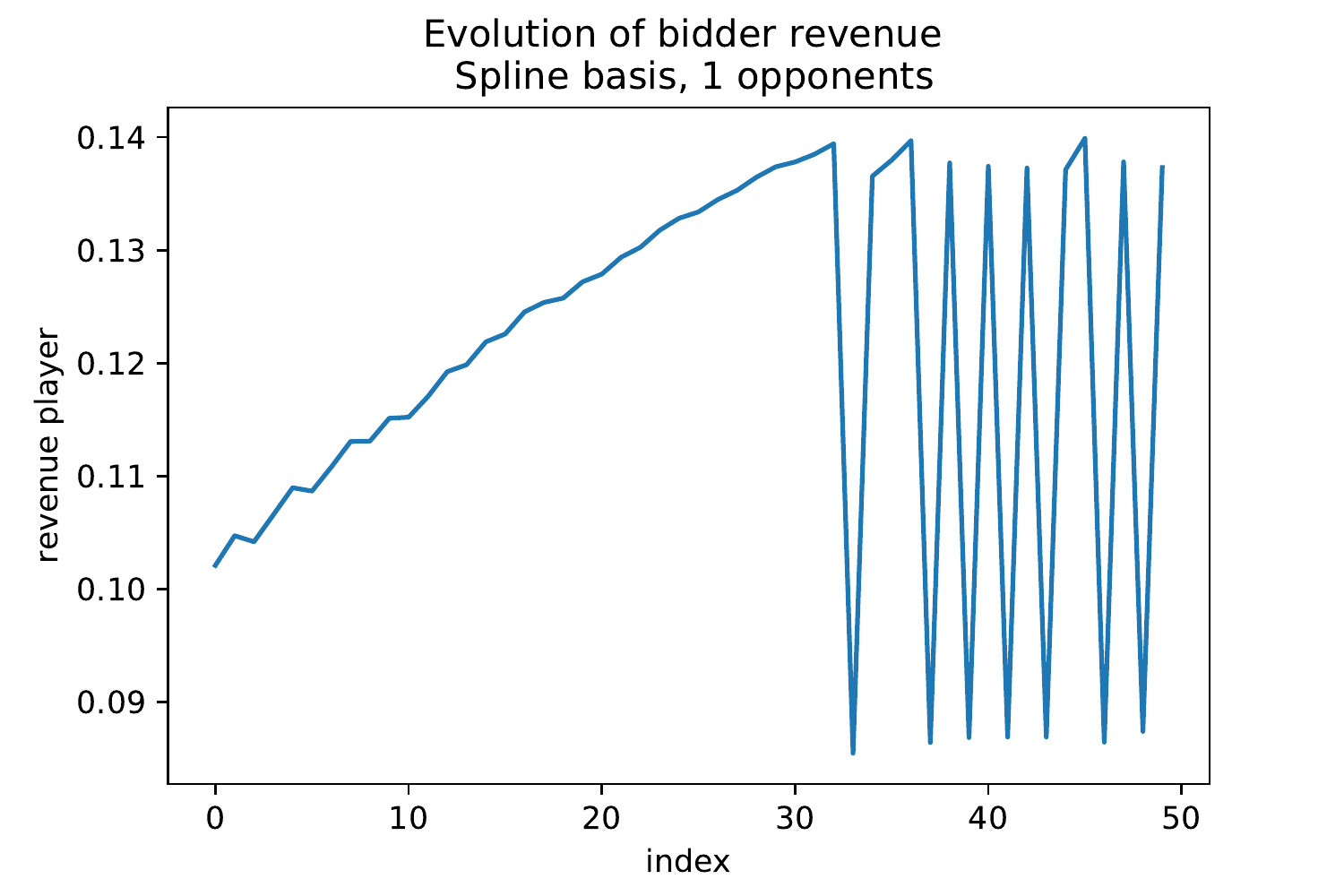}% &
 \end{tabular}
\caption{\textbf{During gradient descent.} The seller revenue is flat making the estimation of the objective depending on the reserve value hard and unstable.}
\label{fig:during_grad_descent}
\end{figure}
This experiment was ran in the case of two bidders with uniform value distribution. The first one is strategic and the second one is bidding truthfully. We consider the case of a lazy second price auction. The reserve price of the first bidder is computed based on her bid distribution. We recall that in the case of two bidders bidding truthfully with uniform distribution and reserve price equal to 0.5, the utility of one bidder is equal to 0.083. We see during the first 100 steps, the optimization works well and the algorithm finds a strategy with utility around 0.12 which is already a 44\% increase compared to the truthful strategy. Then, optimization is very unstable because seller revenue is flat and the computation of the reserve value very unstable. This is why we introduce the relaxation that we optimize in the next section.

Compared to our Pytorch code that we introduce in the next section optimizing the relaxation, the gradient descent code is also very slow because the reserve value is computed by exhaustive search. Having a fully differentiable approach enables also to adapt it easily to many other different auction mechanisms.

\newpage

\section{Numerical results with the relaxation and comments on the different strategies}
\label{appendix_experiments}
\subsection{Comments on the code to reproduce the experiments}
To rerun the code, we advise the careful reader to use the notebook notebookStrategiesPytorch.ipynbto see how to compute the different strategies in the various settings presented in the paper. The code uses Python 3.6 and Pytorch 0.4.1.
\subsection{Comparison between the theory and the results of the optimization for the Myerson auction}
\label{comparisonMyerson}
Figure \ref{fig:myerson_results} should be compared with Figure \ref{fig:uniform_bid_profiles}. Figure \ref{fig:uniform_bid_profiles} gives the optimal strategy with a strategic bidder with uniform distribution against 3 other bidders that are bidding truthfully with also a uniform distribution. Figure \ref{fig:myerson_results}  is the result of the optimization of
$$
U_i(\shadingFunc_i)=\Expb{[X_i-h_{\beta_i}(X_i)] F_Z(h_{\beta_i}(X_i))}\;.
$$
with $F_Z$ the cumulative distribution function of $Z=\max_{2\leq j \leq K}(0,\vValue_j(X_j))$, $X_i$ is the value of bidder $i$, and $h_{\beta_i} = \psi_{B_i}(\beta_i(X_i))$ is the virtual value function associated with the bid distribution.
\begin{figure}[h!]
\begin{tabular}{cccc}
\includegraphics[width=0.33\textwidth, height=0.19\textheight]{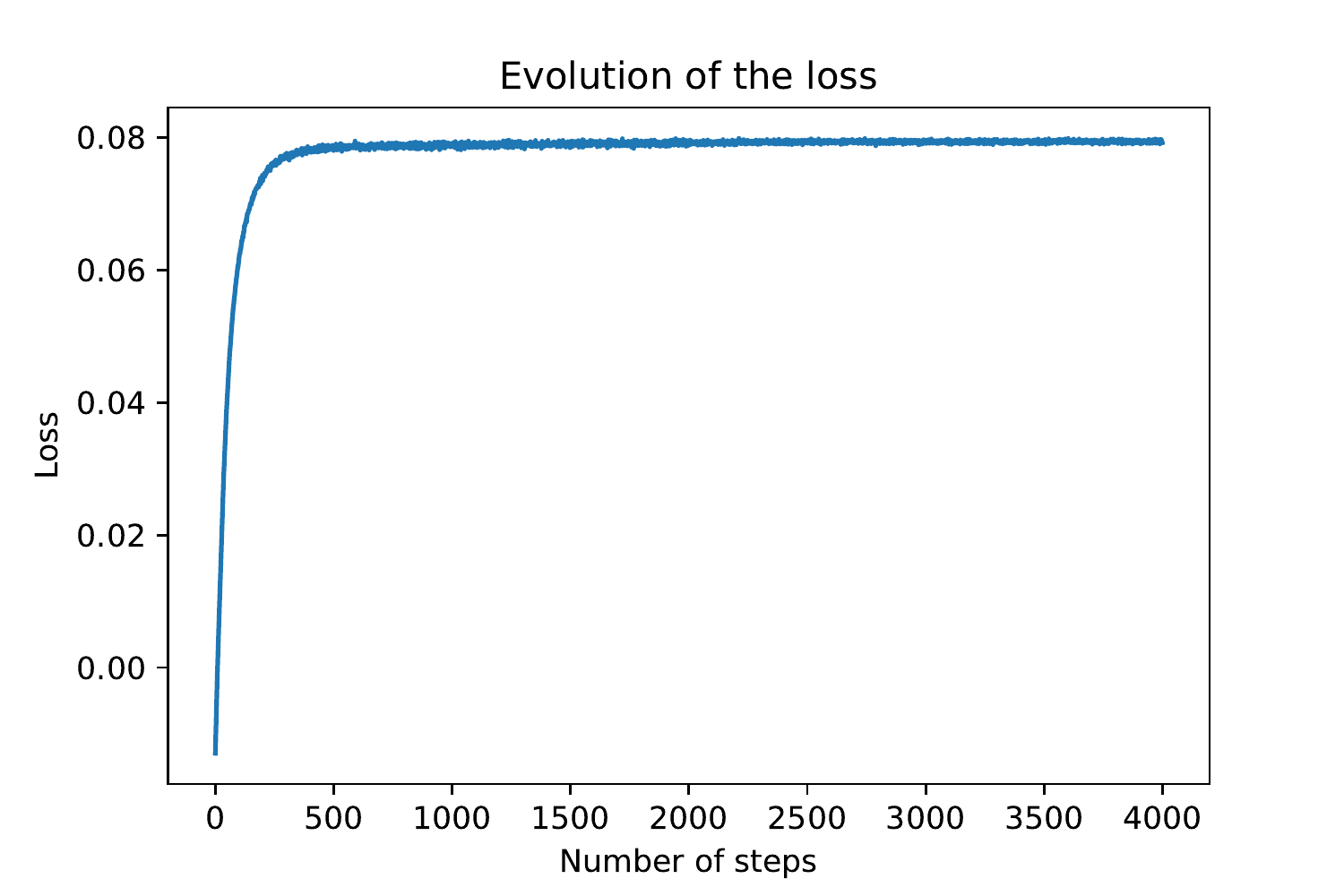} &
\includegraphics[width=0.33\textwidth, height=0.19\textheight]{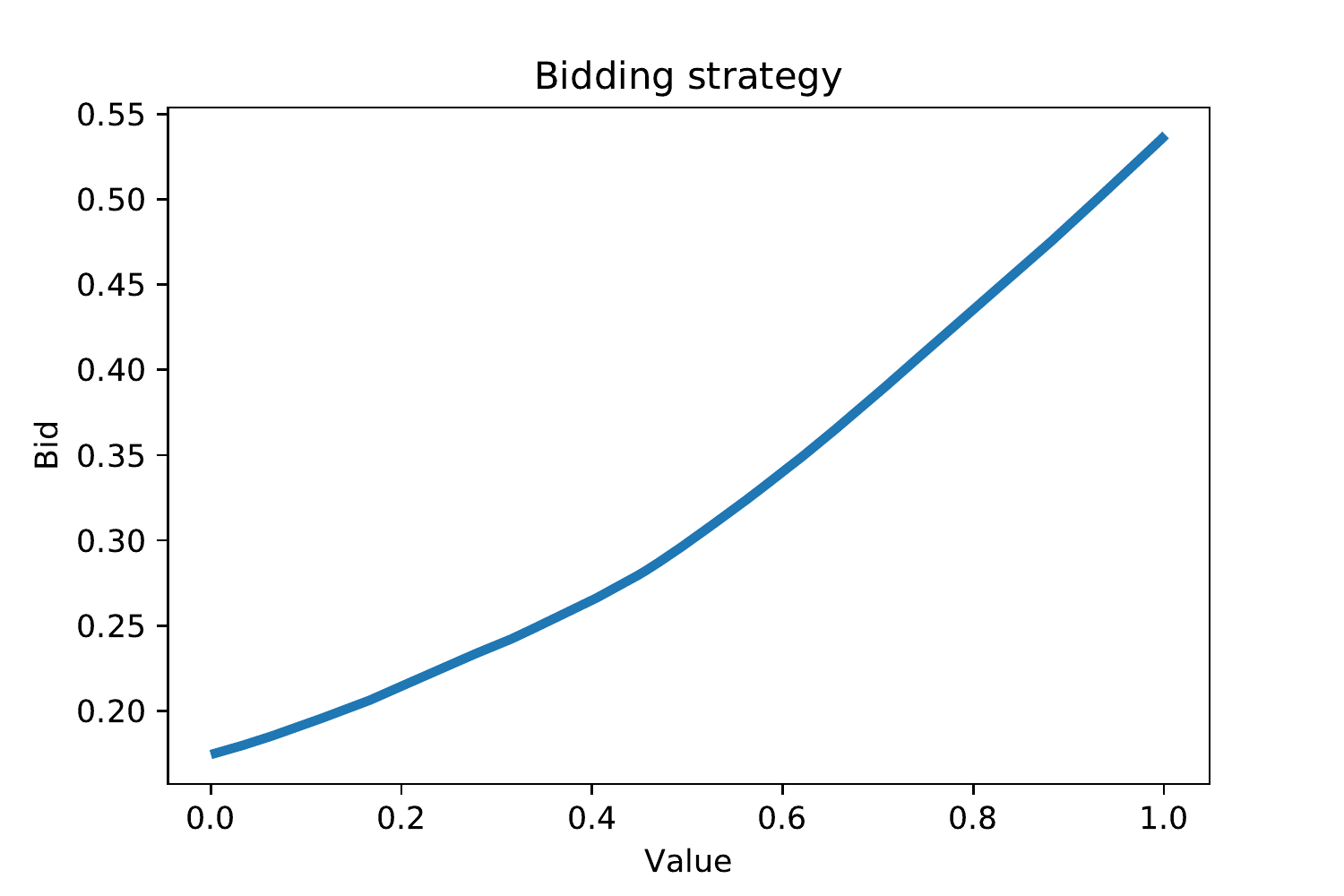}&
\includegraphics[width=0.33\textwidth, height=0.19\textheight]{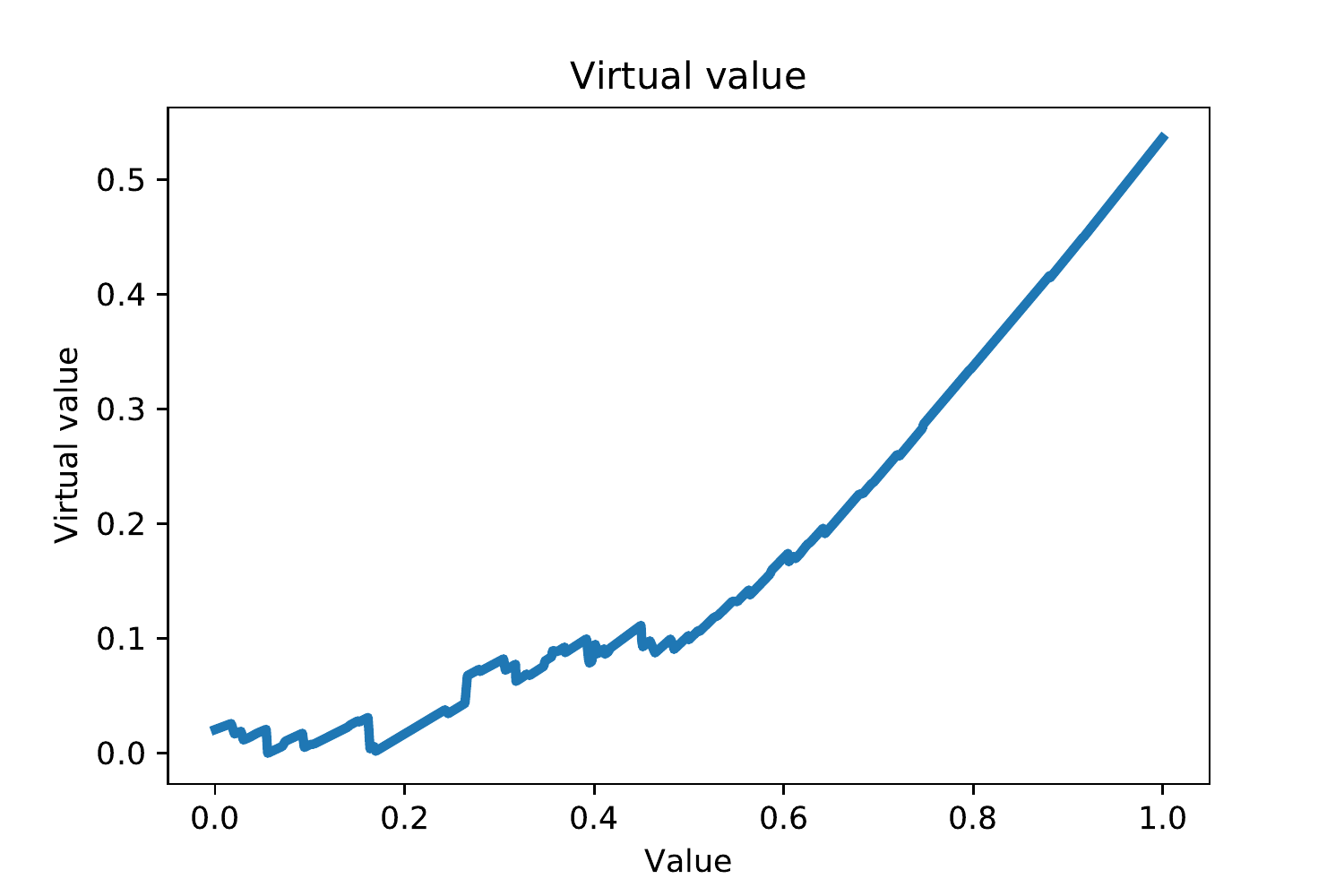}% &
 \end{tabular}
\caption{\textbf{Left:}Evolution of the objective during the learning. \textbf{Middle:} Bidding strategy at the end of the training. \textbf{Right:} Virtual value at the end of the training.The objective correspond to one Myerson auction with three opponents.}
\label{fig:myerson_results}
\end{figure}
\subsection{Results on eager second price auction with monopoly reserve prices}
\label{eagerappendix}

\begin{figure}[h!]
\begin{tabular}{cccc}
\includegraphics[width=0.33\textwidth, height=0.19\textheight]{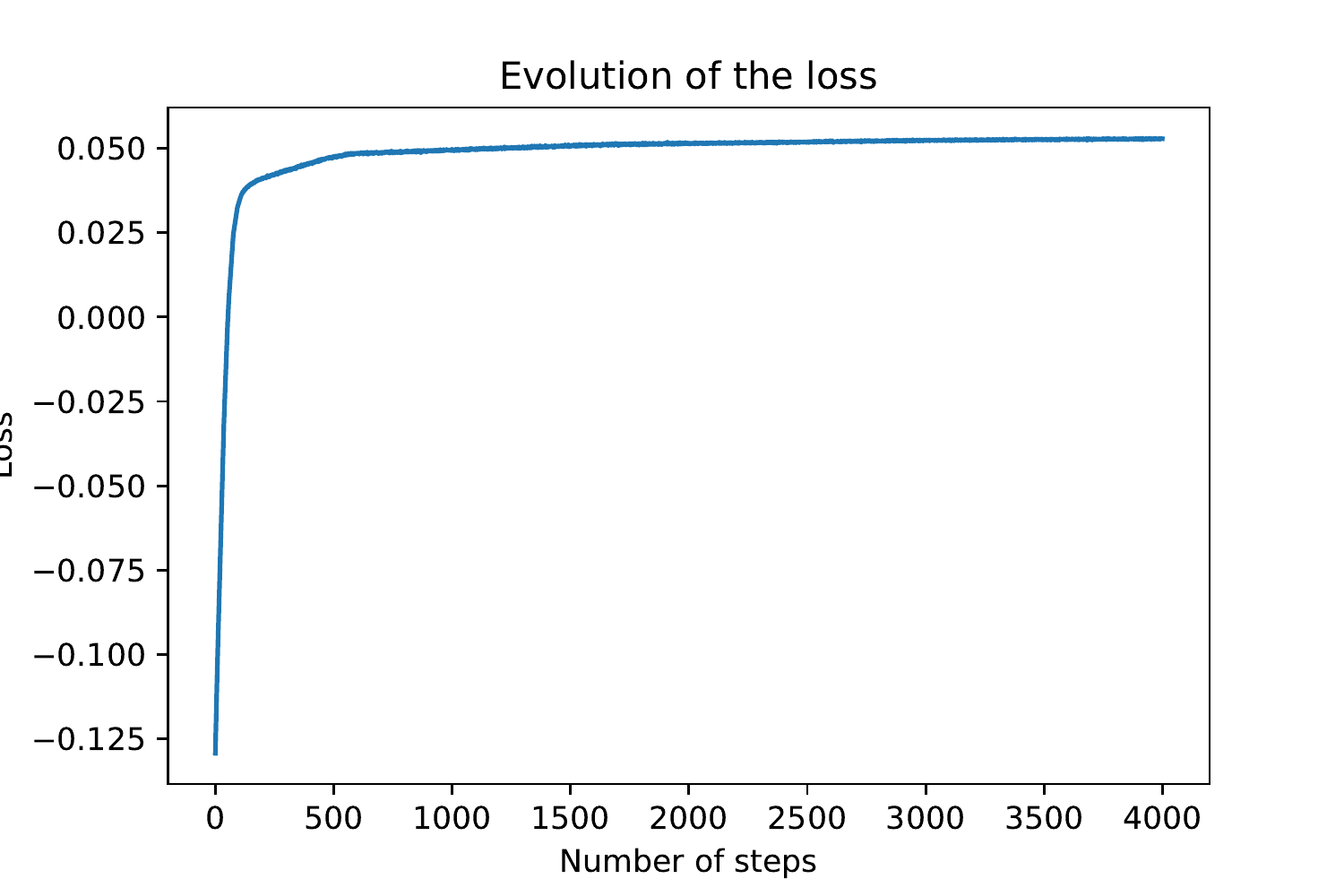} &
\includegraphics[width=0.33\textwidth, height=0.19\textheight]{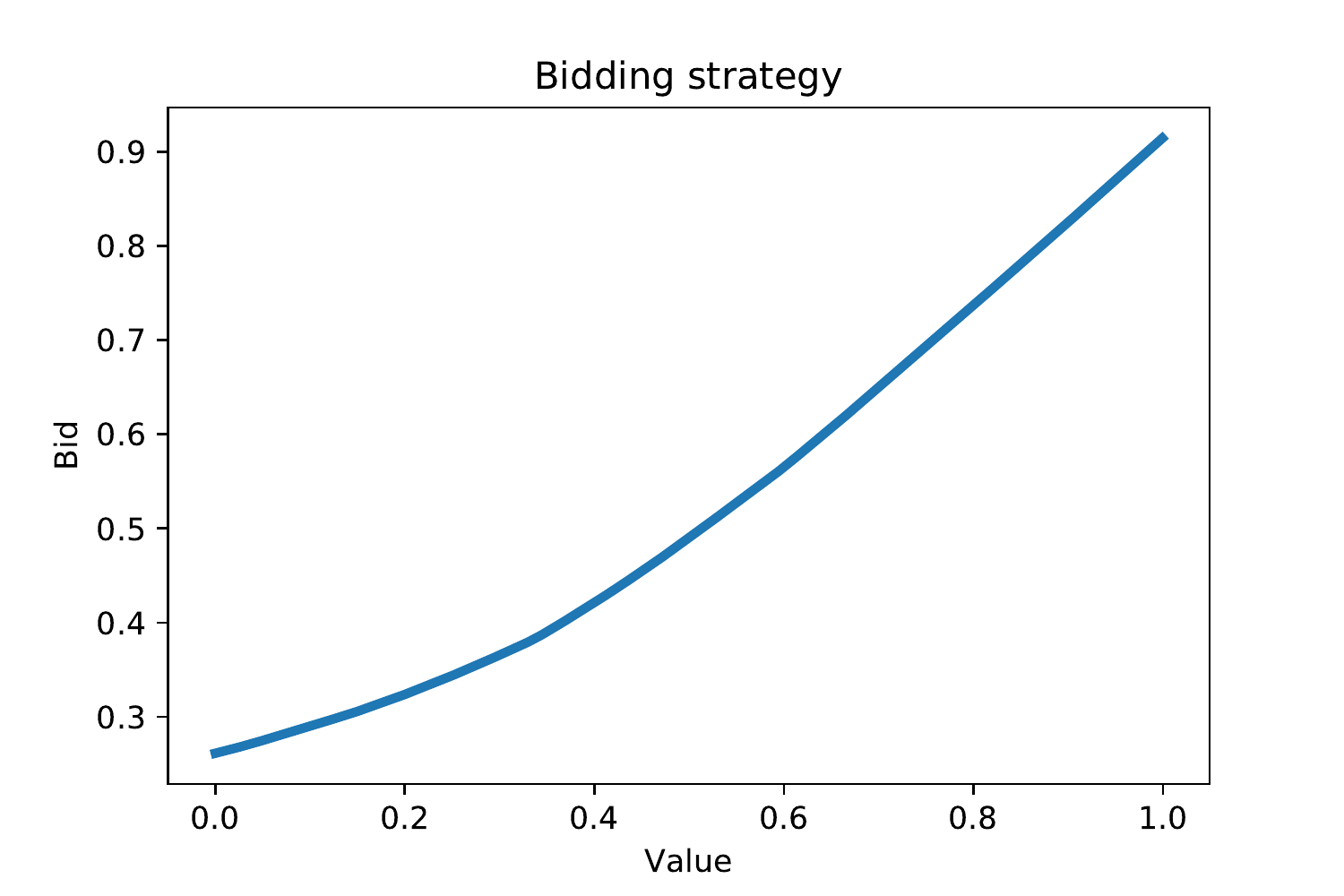}&
\includegraphics[width=0.33\textwidth, height=0.19\textheight]{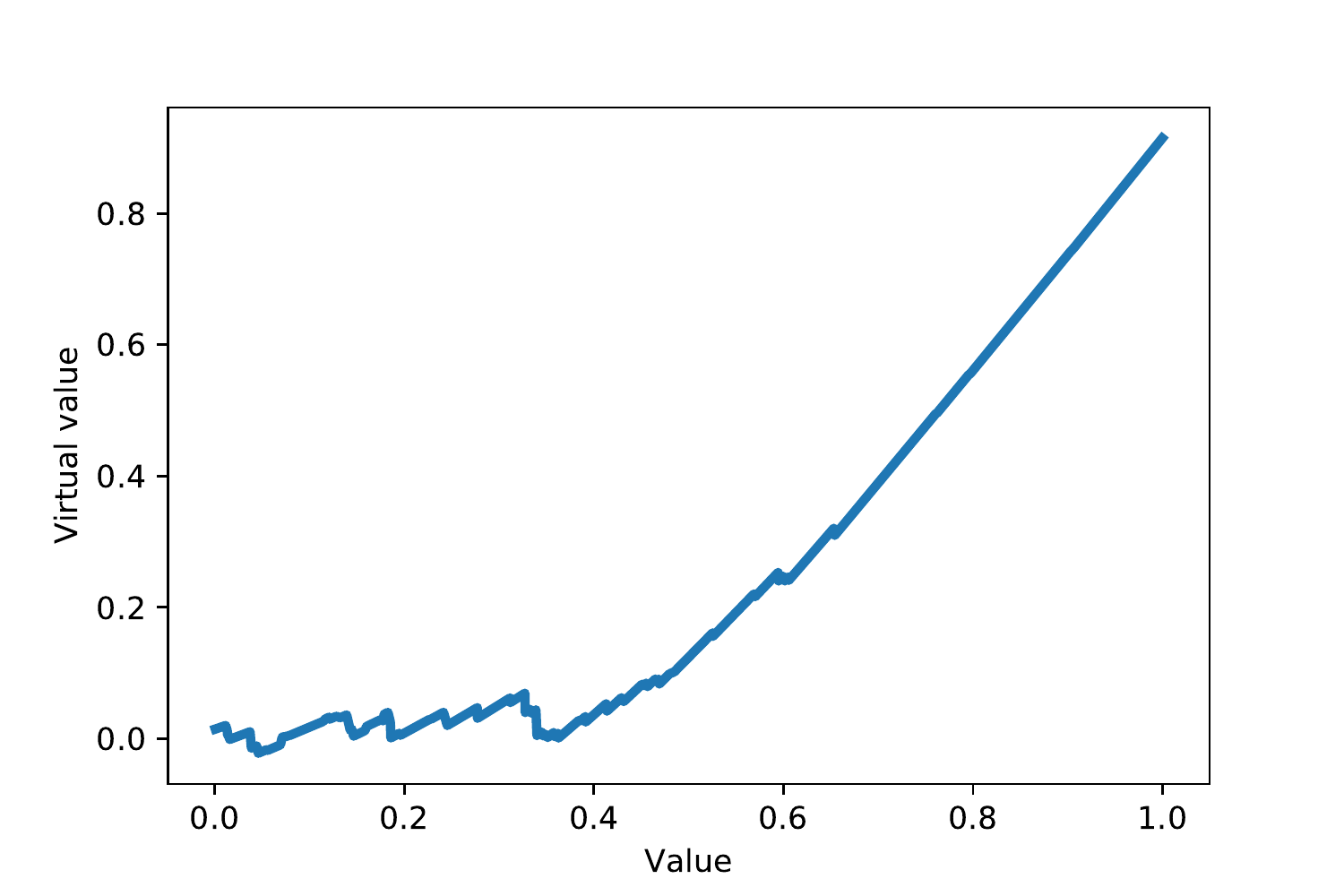}% &
 \end{tabular}
\caption{\textbf{Left:}Evolution of the objective during the learning. \textbf{Middle:} Bidding strategy at the end of the training. \textbf{Right:} Virtual value at the end of the training. The objective correspond to one eager second price auction with three opponents.}
\label{fig:eager_results}
\end{figure}

According to several papers, the eager second price auction with monopoly price is widely used because it leads to higher revenue for the seller in practice when facing asymetric bidders. The eager second price auction consists in running a second price auction but only among bidders that are above their personalized reserve price. The objective function is very similar to the one of the lazy second price auction except that the winning distribution is different below the reserve price of the other bidders. Indeed, if all bidders are below their reserve price, the strategic that bids above his monopoly price is sure to win.

The objective we optimize in our code is the following. If we note $r$ the common reserve price of the other bidders, with $G_i(x)=\prod_{j=2}^K F_{j}(r)\indicator{x\leq r}+\prod_{j=2}^K F_{j}(x)\indicator{x\geq r}$, the objective can be written as
 \begin{equation*}
U_r(\beta_i) = \mathbb{E}\bigg((X_i-h_{\beta_i}(X_i))G_i(\beta(X_i))\textbf{1}(h_{\beta_i}(X_i) \geq 0)\bigg)\;,
\end{equation*}
The following figure was derived with assuming a uniform distribution for all bidders.
\end{document}

%% file: formalProofKKT.tex
%!TEX root = /Users/n.elkaroui/OneDrive - CRITEO/NEKDocs/NotesCriteo/TruthfulAuctionsAreNotTruthful/Strategic_Seller/papers_auction/gradient_descent_paper/gradient_descent.tex

\section{Proofs for Section 3 }
Recall that our setup is that we are given a strategy $\beta$ and a current reserve value $r$. We want to extend our strategy below $r$ in a way that is optimal, at least locally optimal. We assume throughout that the seller is welfare benevolent.

So we have to solve the infinite programming problem 
\begin{gather*}
\max_\beta F(\beta)=\Expb{[X-h_\beta(X)]G(\beta(X))\indicator{0\leq X\leq r}}\\
\text{ subject to }\forall t\;, 0\leq t \leq r \; \Expb{h_\beta(X)G(\beta(X))\indicator{t\leq X\leq r}}\leq 0 \\
\text{ and } \Expb{h_\beta(X)G\beta(X)\indicator{0\leq X\leq r}}= 0\;, \beta(r)=\beta(r^+)  
\end{gather*}
where $\beta(r^+)$ is given by the strategy that had reserve value at $r$. This is just a continuity requirement and it ensures that Myerson's formula applies. 

Of course the seller revenue for bids below $r$ if she sets the reserve value at $t$ is 
$$
\Expb{h_\beta(X)G(\beta(X))\indicator{t\leq X\leq r}}
$$
Note that the constraints mean that the max revenue of the seller is achieved for the reserve value 0 and it is 0. We know that the reserve value should be 0, because otherwise the buyer could use thresholding below the reserve value to increase her revenue and not change the revenue the seller derives from her in a lazy second price auction. So that guarantees that the reserve value is 0. 

\subsection{The case $G(x)=\min(x,1)$}\label{subsec:app:particularCaseOptimality}
Let $r$ be given. Call $-h(x)$ the revenue of the seller on $[x,r]$. The constraints on $h$ is that $h(r)=0$, $h(0)=0$ and $h\geq 0$ on $[0,r]$. That way the revenue of the seller is maximized at $x^*=0$. 

We assume that $\beta\leq \beta(r)=r<1$. So then $G(\beta)=\beta$. 
\subsubsection{Preliminaries}
Using results in the main text, i.e. $h_\beta(x)f(x)=[\beta(x)(F(x)-1)]'$, our constraints can be written in differential form as 
$$
\int_x^r [\beta(F-1)]' G(\beta)=\int_x^r [\beta(F-1)]'\beta= \int_x^r [\beta(F-1)]'\frac{\beta (F-1)}{F-1}=-h\;.
$$
So if $u=[\beta(1-F)]^2$, we have equivalently 
$$
\int_x^r u' \frac{1}{1-F}=2h(x)\;.
$$
If we integrate by parts, using the fact that $(1/(1-F))'=f/(1-F)^2$ and call $m(x)=u(x)/(1-F(x)=\beta^2(1-F)$, we get 
\begin{equation}\label{eq:definitionAuxiliaryFunctionm}
m(r)-m(x)-\int_x^r m(u)\frac{f(u)}{1-F(u)}=2h(x)\;.
\end{equation}
Assuming temporarily $h$ is differentiable, we differentiate Equation \eqref{eq:definitionAuxiliaryFunctionm} to get 
$$
-m'(x)+m(x)g(x)=2h'(x)\;, \text{ with } g(x)=\frac{f(x)}{1-F(x)}=-(\log(1-F))'\;.
$$
Because this is a first order ODE, we can integrate this equation fully to get 
$$
m(x)=\frac{1}{1-F(x)}\left[\beta^2(r)(1-F(r))^2+2\int_x^r h'(t)(1-F(t)) dt \right]\;.
$$
Note that this is the family of solutions among $m$ that are differentiable. 
\subsubsection{Key result}
So we have the following theorem. 
\begin{theorem}
Suppose $G(x)=min(x,1)$ and call $-h$ the revenue of the seller at $x$. The unique shading function $\beta$ such that $\beta^2 (1-F)$ is differentiable and $\beta(x)\leq 1$ is $\beta>0$ such 
$$
\beta^2(x;h)=\beta^2(x)=\frac{1}{(1-F(x))^2} \left[\beta^2(r)(1-F(r))^2+2\left(\int_x^r h(t)f(t)dt - h(x)(1-F(x))\right) \right]\;.
$$		
Recall that our constraint is that $\beta(r)=r<1$. Of course $h=0$ corresponds to the thresholded function. 
\end{theorem}

\begin{proof}
Put all of the above together and then integrate by parts 	$\int_x^r h'(t)(1-F(t)) dt$ to get the formulation above. 
\end{proof}

% \textbf{Remark~:} we see that $\beta^2_h$ depends linearly on $h$. So it suggests that if the optimum for our optim problem was achieved for $h\neq 0$, it would work for $\lambda h$, any $\lambda>0$. So we can do perturbative analysis at $h$ fixed, by looking at $\beta^2(\cdot;\eps h)$.
\subsubsection{Back to the optimization problem}
We can naturally view $h$ as a slack variable. By contrast to classical finite dimensional optimization, our slack variable is a function. 
We wish to maximize 
$$
\Pi(h)=\int_0^r x G(\beta(x;h)) f(x) dx =\int_0^r x \beta(x;h) f(x) dx\;,
$$
as we assume that $\beta(x;h)$ remains below 1. Note that the part of the integral involving $h_\beta$ has now disappeared as we consider functions for which the average payment over $[0,r]$ is zero - that is the sense of our equality constraint. 

Consider $\Pi(\eps h)$, with $\eps$ very small. Since $\sqrt{1+t}=1+\frac{t}{2}$, we see that 
$$
\beta(x;\eps h)=\beta(x;0)+\frac{\eps}{\beta(r)(1-F(r))}\frac{1}{1-F(x)}\left[\int_x^r h(t)f(t)dt-h(x)(1-F(x))\right]\;,
$$
where we have used that $h(r)=0$. So we have 
$$
\Pi(\eps h)=\Pi(0)+\frac{\eps}{\beta(r)(1-F(r))} \int_0^r \frac{xf(x)}{1-F(x)} \left[\int_x^r h(t)f(t)dt-h(x)(1-F(x))\right]dx\;.
$$
The question becomes whether 
$$
\boxed{
\int_0^r \frac{xf(x)}{1-F(x)} \left[\int_x^r h(t)f(t)dt-h(x)(1-F(x))\right]dx\leq 0\;.
}
$$
\begin{lemma}
Suppose that $\psi_F(x)=x-\frac{1-F(x)}{f(x)}\leq 0$ on $[0,r]$. Then, if $h\geq 0$, 
$$
\int_0^r \frac{xf(x)}{1-F(x)} \left[\int_x^r h(t)f(t)dt-h(x)(1-F(x))\right]dx\leq 0
$$	
\end{lemma}
\begin{proof}
We have, using Fubini's theorem,
\begin{align*}
\int_0^r \frac{xf(x)}{1-F(x)} \left[\int_x^r h(t)f(t)dt-h(x)(1-F(x))\right]dx&=-\int_0^r xh(x)f(x)+\int_0^r\int_0^r
\frac{xf(x)}{1-F(x)} h(t)f(t)\indicator{t>x} dx dt \;, \\
&=-\int_0^r xh(x)f(x)+\int_0^r
h(t)f(t) \int_0 ^x\frac{yf(y)}{1-F(y)}dy  dt\;,\\
&=-\int_0^r h(x)f(x)\left[x-\int_0 ^x\frac{yf(y)}{1-F(y)}dy\right]\;.
\end{align*}
Now since $\psi_F(x)=x-\frac{1-F(x)}{f(x)}\leq 0$, we have equivalently
$$
\frac{xf(x)}{1-F(x)}\leq 1\;.
$$
Hence, 
$$
\int_0 ^x\frac{yf(y)}{1-F(y)}dy\leq \int_0^x 1 dy =x\;. 
$$
We conclude that for all $h\geq 0$, 
$$
\int_0^r \frac{xf(x)}{1-F(x)} \left[\int_x^r h(t)f(t)dt-h(x)(1-F(x))\right]dx\leq 0\;.
$$
\end{proof}

We have the following theorem.
\begin{theorem}\label{thm}
For the problem of revenue optimization of the buyer, thresholding is a locally optimal (among functions such that $\beta^2(1-F)$ is differentiable and $\beta$ is bounded by 1 on $[0,r]$ when the competition has cdf $G(x)=\min(x,1)$ for $x\geq 0$. 

Furthermore, it is globally optimal among those bidding strategies.
\end{theorem}

\begin{proof}
The local optimally comes from the previous analysis. 

The global optimality comes from the fact that using the concavity of $\sqrt{1+x}$, we have $\sqrt{1+x}\leq 1+\frac{x}{2}$ for all $x\geq -1$. Then the upper bound is what we computed above. So an upper bound on the revenue is lower than the thresholded revenue and we have global optimality. 	
\end{proof}

\subsection{The general case, local optimality}\label{subsec:app:generalCaseOptimality}
We now prove local optimality for general differentiable $G$ using differential ideas similar to the ones above.  
Let $\eps>0$ and $\beta_\eps=\beta^*(1+\eps \rho)$. $\beta^*=C/(1-F(x))$ is the shading function corresponding to thresholding the virtual value. Recall that $C=r(1-F(r))$ for continuity at $r$. Below, $\rho$ is a differentiable function. We require $\rho(r)=0$ to have continuity of $\beta_\eps$ at $r$. 
Recall the constraints that for all feasible $\beta$'s, 
$$
\forall x\leq r\;,\;  \int_x^r [\beta(F-1)]' G(\beta)=-h(x)\leq 0 \;, \text{ and } h(0)=0\;.
$$
$h(r)=0$ by construction. We assume that $\beta_\eps$ is feasible for $\eps \in [0,\eps_0)$, which gives us the notion of locality we need. Because $[\beta^*(F-1)]'=0$, we have 
$$
\int_x^r [\beta_\eps(F-1)]' G(\beta_\eps)=\eps \int_x^r (\rho\beta^*(F-1))' G(\beta^*+\eps \rho\beta^*)=-h_\eps
$$
and hence by limiting ourselves to the first order term in the Taylor expansion (the second order term in $\eps$ is asymptotically negligible), 
$$
\int_x^r (\rho\beta^*(F-1))' G(\beta^*)=-h(x)\leq 0\;, \forall x\leq r\;.
$$
Assuming that $h$ is differentiable, we can differentiate the previous equality to get 
$$
(\rho(x)\beta^*(x)(F-1)(x))' G(\beta^*(x))=h'(x)\;,
$$
and using the fact that $\rho(r)=0$ because we want continuity at $r$, we get 
$$
\beta^*(x)\rho(x)(1-F(x))=\int_x^r \frac{h'(x)}{G(\beta^*(x))}dx=\int_x^r \frac{h(x)[\beta^{*}]'(x)g(\beta^*(x))}{G^2(\beta^*(x))}dx -\frac{h(x)}{G(\beta^*(x))}\;.
$$
The last equality comes from integration by parts. 

\subsubsection{Back to the optimization problem}
Recall that we see to maximize over admissible $\rho$'s (i.e. those for which the inequality constraints are verified), 
$$
\max_\rho \int_0^x x f(x) G(\beta^*(x)(1+\eps \rho(x))) dx\simeq \max_\rho \int_0^x x f(x) G(\beta^*(x))+ 
\eps \int_0^r x f(x) g(\beta^*(x)) \rho(x) \beta^*(x) dx\;.
$$
Where the equality is to first order in the Taylor expansion. 
We claim that if $\rho$ is admissible, then the second term is negative. 
\begin{lemma}\label{app:lemma:decreaseInBuyerUtilityAwayFromThresholding}
Let $\beta^*(x)=C/(1-F(x)$. Assume that $G(\beta^*(x))>0$ on [0,r]. Suppose that $\psi_F(x)\leq 0$ on $[0,r)$. Suppose the function $\rho$ is such that for $h\geq 0$, $h(r)=0$
$$
\beta^*(x)\rho(x)(1-F(x))=\int_x^r \frac{h(x)[\beta^{*}]'(x)g(\beta^*(x))}{G^2(\beta^*(x))}dx -\frac{h(x)}{G(\beta^*(x))}\;.
$$
Then 
$$
I=\int_0^r x f(x) g(\beta^*(x)) \rho(x) \beta^*(x) dx\leq 0\;.
$$
\end{lemma}
\begin{proof}
The strategy is the same as above. Let us call $\beta^*=\beta$ to make the notation less cumbersome. 
Replacing $\rho$ by its value, we have, using the fact that $\beta(x)=C/(1-F(x))$,
\begin{align*}
I/C&=-\int \frac{xf(x)}{1-F(x)}g(\beta(x)) \frac{h(x)}{G(\beta(x))} dx+\int_0^r dx \frac{xf(x)}{1-F(x)}g(\beta(x))\int_x^r \frac{h(t)\beta'(t)g(\beta(t))}{G^2(\beta(t))}dt\\
&=-\int \frac{xf(x)}{1-F(x)}g(\beta(x)) \frac{h(x)}{G(\beta(x))} dx+
\int_0^r\int_0^r \frac{xf(x)}{1-F(x)}g(\beta(x))\frac{h(t)\beta'(t)g(\beta(t))}{G^2(\beta(t))} \indicator{t\geq x} dx dt\\
&=-\int \frac{xf(x)}{1-F(x)}g(\beta(x)) \frac{h(x)}{G(\beta(x))} dx+
\int_0^r dx \frac{h(x)\beta'(x)g(\beta(x))}{G^2(\beta(x))} \int_0^x \frac{yf(y)}{1-F(y)}g(\beta(y))dy
\end{align*}
Now recall that $\beta'(t)=Cf/(1-F)^2$, $C>0$,  so that 
\begin{align*}
\int_0^x \frac{yf(y)}{1-F(y)}g(\beta(y))dy&=\frac{1}{C} \int_0^x y(1-F(y))\beta'(y)g(\beta(y))dy\;,\\
&=\frac{1}{C}\left[\left.t(1-F(t))G(\beta(t))\right|_0^x-\int_0^xG(\beta(t))[t(1-F(t))]' \right]\;.
\end{align*}
Of course, $(t(1-F(t)))'=-f(x)\psi_F(x)\geq 0$ since $\psi_F(x)\leq 0$. So we have 
$$
\int_0^x \frac{yf(y)}{1-F(y)}g(\beta(y))dy\leq \frac{1}{C} x (1-F(x)) G(\beta(x))\;.
$$
Hence, since $\beta'(x)>0$ and $h(x)>0$
\begin{align*}
\int_0^r dx \frac{h(x)\beta'(x)g(\beta(x))}{G^2(\beta(x))} \int_0^x \frac{yf(y)}{1-F(y)}g(\beta(y))dy
&\leq 
\int_0^r dx \frac{h(x)\beta'(x)g(\beta(x))}{G^2(\beta(x))} \frac{1}{C} x (1-F(x)) G(\beta(x))\\
&=\int_0^r dx \frac{h(x)\frac{Cf(x)}{[1-F(x)]^2}g(\beta(x))}{G(\beta(x))} \frac{1}{C} x (1-F(x))\\
&=\int_0^r dx \frac{xh(x)\frac{f(x)}{[1-F(x)]}g(\beta(x))}{G(\beta(x))}
\end{align*}
We conclude that 
$$
I\leq 0\;.
$$
\end{proof}
We have the following theorem. 
\begin{theorem} \label{thm:localOptimalityGeneralCase}
Suppose that $F$, $1-F$ and G are differentiable. Suppose furthermore that $\psi_F(x)=x-(1-F(x))/f(x)\leq 0$ on $[0,r]$. Then the function $\beta^*=r(1-F(r))/(1-F(x))$ is locally optimal among differentiable shading strategies, provided $G(\beta^*(x))>0$ on [0,r]. 
\end{theorem}
\begin{proof}
Two cases are possible. Either $\beta^*$ is an isolated point in which case it is by definition locally optimal. If that is not the case, the previous computations show that $\beta^*$ is locally optimal as any local deviation in the feasible set of shading functions yields lower utility for the buyer - that is the content of Lemma \ref{app:lemma:decreaseInBuyerUtilityAwayFromThresholding}.
\end{proof}